\documentclass{article}
\usepackage{amsmath}
\usepackage{amsthm}
\usepackage{amssymb}
\usepackage{cases}
\usepackage{bm}
\usepackage{breqn}
\usepackage{amsfonts}
\usepackage{color}
\usepackage[usenames,dvipsnames]{xcolor}

\usepackage[boxsize=0.3em,aligntableaux=center]{ytableau}
\newcommand{\yd}{\ydiagram}

\usepackage{multirow}
\usepackage{booktabs} % \cmidrule
\usepackage{array}
\newcolumntype{L}{>{$}l<{$}} % math-mode version of "l" column type
\newcolumntype{R}{>{$}r<{$}} % math-mode version of "r" column type
\newcolumntype{C}{>{$}c<{$}} % math-mode version of "c" column type

\makeatletter
\renewcommand*\env@matrix[1][*\c@MaxMatrixCols c]{%
  \hskip -\arraycolsep
  \let\@ifnextchar\new@ifnextchar
  \array{#1}}
\makeatother

\newtheorem{theorem}{Theorem}[section]
\newtheorem{lemma}[theorem]{Lemma}
\newtheorem{proposition}[theorem]{Proposition}
\newtheorem{corollary}[theorem]{Corollary}
\newtheorem{remark}{Remark}[section]

\usepackage{tikz}
\usetikzlibrary{shapes,arrows,matrix}

\usepackage[colorlinks=true,linkcolor=blue,citecolor=red,urlcolor=Violet,bookmarksdepth=subsection,final]{hyperref}

  \newcommand{\del}{\partial}
  
  \newcommand{\eps}{\varepsilon}
  
\renewcommand{\d}{\mathrm{d}}
  \newcommand{\K}{\mathrm{K}}
  \newcommand{\CK}{\mathrm{CK}}
  \newcommand{\CYK}{\mathrm{CYK}}
  \newcommand{\CKY}{\CYK}

  \newcommand{\bd}{\overline{\d}}  
    
  \newcommand{\bCK}{\overline{\CK}}
  
  \newcommand{\bCKY}{\overline{\CKY}}
  
\DeclareFlexCompoundSymbol{\coloneq}{Rel}{:=}
\DeclareFlexCompoundSymbol{\eqcolon}{Rel}{=:}
  
\title{Closed Conformal Killing-Yano Initial Data}
\author{%
	Alfonso Garc\'{\i}a-Parrado$^{\sharp}$%
		\thanks{E-mail: alfonso@utf.mff.cuni.cz} {} and
	Igor Khavkine$^\flat$%
		\thanks{E-mail: khavkine@math.cas.cz}
\\[2ex]
	{\small $^\sharp$Institute of Theoretical Physics, Faculty of Mathematics and Physics,}\\
	{\small Charles University in Prague, V~Hole\v{s}ovi\v{c}k\'ach~2, 180~00 Praha 8, Czech Republic}\\
	{\small $^\flat$Institute of Mathematics of the Czech Academy of Sciences,}\\
	{\small \v{Z}itn{\'a} 25, 115 67 Praha 1, Czech Republic}
}
\date{}

\begin{document}
\maketitle

\begin{abstract}
Through an exhaustive search, we produce a 5-parameter family of
propagation identities for the \emph{closed conformal Killing-Yano}
equation on 2-forms, which hold on an Einstein cosmological vacuum
spacetime in any dimension $n>4$. It is well-known that spacetimes
admitting a non-degenerate 2-form of this type are exhausted by the
Kerr-NUT-(A)dS family of exact higher dimensional black hole solutions.
As a consequence, we identify a set of necessary and sufficient
conditions ensuring that the cosmological vacuum development of an
initial data set for Einstein's field equations admits a closed
conformal Killing-Yano 2-form. We refer to these conditions as
\emph{closed conformal Killing-Yano initial data} (cCYKID) equations.
The 4-dimensional case is special and is treated separately, where we
can also handle the conformal Killing-Yano equation without the closed
condition.
\end{abstract}

\section{Introduction} \label{sec:intro}

Solutions of the \emph{Killing} equation are vector fields generating
infinitesimal isometries of Lorentzian spacetimes or more generally
(pseudo-)Riemannian geometries. Generalizations of the Killing equation
to higher rank tensors~\cite{cariglia} include the
\emph{Killing-St\"ackel}~\cite{eisenhart1934} equations on symmetric
tensors, as well as the \emph{Killing-Yano} equations on
$p$-forms~\cite{yano1952}. Solutions of these equations, the higher rank
Killing tensors, can be associated with so-called hidden symmetries,
which are responsible for the integrability of geodesic equation, and
the separability of Hamilton-Jacobi or wave/Laplace equations, as well
as supersymmetric or spinorial generalizations of any of these
equations~\cite{cariglia, grvh-ky, tanimoto-ky}. Of particular interest
is the closely related equation for \emph{closed conformal Killing-Yano
2-form}, which is responsible for the complete integrability of
Einstein's equations, resulting in the so-called \emph{Kerr-NUT-(A)dS}
family of higher dimensional rotating black holes (where this 2-form is
called the \emph{principal tensor}). This result has by
now a substantial literature, with~\cite{Houri_2007a, Houri_2007b,
Houri_2012, Houri_2015, sergyeyev-krtous, kfk-kerrnutads} being some key
references and more listed in the extensive review~\cite{fkk-review},
and it has motivated us to focus on the equation for closed conformal
Killing-Yano 2-forms. For our work, the most relevant aspects of this
equation are its integrability conditions, which are conveniently
summarized in~\cite{Batista2015}. Later in this work, we make use of
some representation theoretic methods and Young diagrams. These tools
have also been recently fruitfully used to study the integrability
conditions of higher rank Killing-St\"ackel and Killing-Yano
tensors~\cite{Houri_2018}, a development that goes in a different
direction than our work.

In the recent work~\cite{gpkh-ckid} we returned to the question of how
to detect the presence of solutions of a geometric PDE on the bulk of a
solution of Einstein's equations just by looking at the initial data for
the metric. This question had been studied and successfully answered, by
deriving the corresponding initial data equations, for only for a small
number of examples: Killing equation~\cite{berezdivin, MONCRIEF-KID1,
MONCRIEF-KID2, COLL77, beig-chrusciel}, homothetic Killing
equation~\cite{berger}, and some Killing spinor
equations~\cite{GOMEZLOBO2008, bk-kerrness1, bk-kerrness2, gasperin-williams} (though, only in 4
dimensions). In~\cite{gpkh-ckid} we have succeeded in adding the
conformal Killing equation to this short list. In this work,
motivated by the possibility of characterizing the initial data giving
rise to the Kerr-NUT-(A)dS family of spacetimes, we use the methods
of~\cite{gpkh-ckid} to derive the initial data equations for closed
conformal Killing-Yano (cCYK) 2-forms. Such an initial data
characterization is complementary to the above mentioned local
characterization by the existence of a principal tensor. In particular,
such an initial data characterization would be independent of the way
the corresponding Cauchy surface is embedded in the ambient spacetime
(cf.~Remark~\ref{rem:isometric-embedding}). In numerical relativity, a
Cauchy surface independent characterization of Kerr initial data, which
was constructed by one of us~\cite{gpgl-kerr} (though using a different
approach), has already been fruitfully exploited~\cite{kerrness-numeric}
to quantitatively estimate the convergence of a ringdown simulation to a
Kerr background. Our results might lead in the future to similar
applications in numerical relativity in higher dimensions, or in the study
of non-linear stability of Kerr-NUT-(A)dS black holes in mathematical
relativity.

In Section~\ref{sec:propeq} we recall the general strategy
from~\cite{gpkh-ckid}, setup the notation, and recall the simplest
example of the Killing initial data (KID). The strategy involves
identifying a propagation identity (Proposition~\ref{prp:propeq}), whose
existence is then responsible for the successful identification of the
desired initial data conditions. At the moment, such a propagation
identity can only be found by trial and error, or by an exhaustive
search. In Section~\ref{sec:ccyk}, we use representation-theoretic ideas
to carry out an exhaustive search, at low differential order, for a
propagation identity for cCYK 2-forms. Some of the more technical
details are relegated to Appendices~\ref{app:tensors}
and~\ref{app:hyp-conds}. The search is successful in all spacetime
dimensions higher than $n=4$ and yields a multi-parameter family of
propagation identities (Theorem~\ref{thm:ccyk-propeq}). The $n=4$ case
is handled separately (Theorem~\ref{thm:ccyk-propeq-4dim}), where the
needed propagation identity was discovered after some trial and
error. The corresponding initial data conditions are derived in
Theorems~\ref{thm:ccykid} and~\ref{thm:ccykid-4dim}, respectively.
Finally, in Section~\ref{sec:cykid}, we adapt our methods to conformal
Killing-Yano 2-forms (without the closed condition) in $n=4$ dimensions.
The equivalent spinorial result was first obtained
in~\cite{GOMEZLOBO2008}, but we give a purely tensorial result and
derivation, which are in line with our motivation to improve the
characterization of the initial data of 4-dimensional rotating black
holes~\cite{gpgl-kerr}.

While we concentrate on Lorentzian geometries, all the covariant
identities that we present are valid also for pseudo-Riemannian
geometries of any signature. All the computations of this paper have
been double-checked with the tensor computer algebra system {\em
xAct}~\cite{XACT,XPERM}.

\section{Propagation equations and initial data characterizations} \label{sec:propeq}

In this section, we give the necessary background information for
presenting our new results in Sections~\ref{sec:ccyk}
and~\ref{sec:cykid}. Namely, we describe what we mean by an
\emph{initial data system} or \emph{initial data characterization} and
state the main proposition about \emph{propagation identities}
(Proposition~\ref{prp:propeq}), which defines the parameters of the
exhaustive search we will later perform to find the initial data systems
for (closed) conformal Killing-Yano equations. A propagation identity,
provided it exists, allows us to conclude that a given auxiliary
condition can be propagated to the future or past if it is satisfied on
initial data. The initial data characterization for this condition then
reduces to the search for a corresponding propagation identity. We have
simply formalized a well-known argument: if some dynamical fields obey
an evolution (or \emph{propagation}) equation, then an auxiliary
property of the initial data (expressed as a differential operator) is
preserved by the evolution if this property itself is propagated by a
compatible evolution equation. What we call a propagation identity
simply captures this compatibility. A classic non-trivial example is the
fact that the harmonic gauge (a.k.a\ wave gauge) condition on initial
data for the Einstein equations is preserved by evolution via the
Einstein equations in harmonic gauge, which was used in the proof of
local well-posedness of the Cauchy problem in General
Relativity~\cite{foures-bruhat-1952}.

Often, it is the main dynamical evolution equation that is fixed and one
searches for convenient auxiliary conditions on the fields that have
compatible propagation equations. Instead, we flip the attention to what
would be the auxiliary condition 
(we call it the \emph{target geometric PDE} and consider it fixed), 
and then look for compatible evolution
equations. Simplifying to the case where all equations are linear, a
propagation identity can be illustrated as follows:
\begin{center}
\begin{tikzpicture}
	\matrix (M) [matrix of math nodes,
			row sep=3em, column sep = 4em, minimum width=2em] {
		\text{\em fields }      & \text{\em eqs }      \\
		\text{\em fields } \phi & \text{\em eqs } \psi \\
	};
	\path[-stealth]
		(M-2-1) edge node [left] {$Q$} (M-1-1)
		        edge node [below] {$E$} (M-2-2)
		(M-2-2) edge node [right] {$P$} (M-1-2)
		(M-1-1) edge node [above] {$\sigma$} (M-1-2);
\end{tikzpicture}
\end{center}
Here our target geometric PDE is $E[\phi]=0$, where $\phi$ denotes a set of fields.
The operator $Q$ evolves the fields, the operator $P$ evolves the
target geometric equation components, collectively labelled by $\psi=E[\phi]$, and the operator
$\sigma$ must exist to close the desired propagation identity:
$P[E[\phi]] = \sigma[Q[\phi]]$. The key implication of this identity is
that when $\phi$ is propagated by $Q[\phi] = 0$, then the value of $\psi
= E[\phi]$ is propagated by $P[\psi] = 0$. In particular, when $\phi$ is
on-shell ($Q[\phi] = 0$) the vanishing of $\psi = E[\phi]$ 
on the initial data hypersurface implies that $E[\phi] = 0$ everywhere (in the domain of
dependence of the initial data). Since we would like the above
propagation identity to be valid for arbitrary solutions of $E[\phi] =
0$, the propagation operator for the fields must itself be a consequence
of the geometric PDE, that is, we need the extra condition 
$Q[\phi] = \rho[E[\phi]]$ for some linear operator $\rho$. 
Finally, since we are working on vacuum backgrounds,
our identities are allowed error terms that vanish when the Einstein
equations $G[g] = 0$ are satisfied. This is merely a simplifying
assumption and there are examples where vacuum propagation identities
can be extended to non-vacuum backgrounds~\cite{racz-kid1, racz-kid2}.
This concludes the explanation of the notation used in the formal
statement of Proposition~\ref{prp:propeq} below.

From now on, all of our differential operators are presumed to be
defined between vector bundles over a manifold $M$ and have smooth
coefficients. What follows is a summary of a more extensive discussion
from our previous work~\cite{gpkh-ckid}.

We call a linear partial differential equation (PDE) $P[\psi]=0$ a
\emph{propagation equation (of order $k\ge 1$)} if it has a well-posed
initial value problem: given a Cauchy surface $\Sigma \subset M$ with
unit normal $n^a$, the equation can be put into Cauchy-Kovalevskaya form
(solved for the highest time derivative) and for each assignment of arbitrary
smooth initial data ${\psi|_{\Sigma}}=\psi_0, \ldots , \nabla_n^{k-1}
{\psi|_{\Sigma}}=\psi_{k-1}$ (where $\nabla_n = n^a\nabla_a$) there
exists a unique solution of $P[\psi]=0$ on all of $M$. In
particular, due to the linearity of the propagation equation, if the
initial data all vanish, $\psi_0 = \cdots = \psi_{k-1} = 0$, then
$\psi=0$ is the corresponding unique solution on $M$.

There are multiple examples of propagation equations: (a) Wave (a.k.a\ 
\emph{normally-hyperbolic}) equations, $P[\psi] = \square \psi +
P'(\nabla\psi,\psi)$~\cite{bgp}, and generalized versions of
those~\cite[Sec.2]{gpkh-ckid}. (b) Transport equations, $P[\psi] = u^a
\nabla_a \psi + P'(\psi)$, with $u^a$ everywhere transverse to
$\Sigma$~\cite{IONESCU-KLAINERMAN-K}. (c) Special cases, like
$P_{bcd}[\psi] = \nabla^a \psi_{abcd}$ for $\psi_{abcd}$ satisfying the
symmetry and tracelessness conditions of the Weyl tensor in
4~dimensions (a so-called \emph{Weyl candidate}~\cite{IONESCU-KLAINERMAN-K}).

\begin{proposition}[{\cite[Lem.1--2]{gpkh-ckid}}] \label{prp:propeq}
Consider a globally hyperbolic spacetime $(M,g)$, satisfying the
Einstein $\Lambda$-vacuum equations, $G_{ab} = R_{ab} - \frac{1}{2} R g_{ab} +
\Lambda g_{ab} = 0$.
Let $E[\phi] = 0$ be a PDE (system) defined on some (possibly
multicomponent) field $\phi$. Suppose that there exist propagation
equations $P[\psi] = 0$, $Q[\phi] = 0$ (of respective orders $k$ and
$l$), where the differential operators $P$ and $Q$ satisfy the identities
\begin{dmath} \label{eq:propeq}
	P[E[\phi]] = \sigma[Q[\phi]] + \tau_P[G] , \quad
	Q[\phi] \hiderel{=} \rho[E[\phi]] + \tau_Q[G]
\end{dmath}
for some linear differential operators $\rho$, $\sigma$, $\tau_P$ and
$\tau_Q$. Then, given a Cauchy surface $\Sigma \subset M$ with unit
timelike normal $n^a$, there is a bijection between the solutions of
$E[\phi] = 0$ and the solutions of $Q[\phi] = 0$ whose initial data
${\phi|_\Sigma} = \phi_0, \ldots, \nabla_n^{l-1} {\phi|_\Sigma} =
\phi_{l-1}$ satisfies ${\psi|_{\Sigma}}=0, \ldots , \nabla_n^{k-1}
{\psi|_{\Sigma}}=0$, for $\psi = E[\phi]$.

In addition, there exists a purely spatial linear PDE on $\Sigma$,
$E^\Sigma[\phi_0,\ldots,\phi_{l-1}] = 0$ such that the conditions
$Q[\phi] = 0$ and $E^\Sigma[{\phi|_\Sigma}, \ldots,
\nabla_n^{l-1}{\phi|_\Sigma}] = 0$ are equivalent to the vanishing of
the initial data ${\psi|_{\Sigma}}=0, \ldots , \nabla_n^{k-1}
{\psi|_{\Sigma}}=0$ for $\psi = E[\phi]$.
\end{proposition}

We must emphasize that the existence of a propagation identity
like~\eqref{eq:propeq} for a particular equation $E[\phi]=0$ is not a
given and must be discovered either by trial-and-error or through a
systematic search. The bulk of this work, in Sections~\ref{sec:ccyk}
and~\ref{sec:cykid}, is devoted exactly to such a systematic search.
When a propagation identity exists, there is no reason for it to be
unique. For instance, we will find such identities in continuous
families.

\begin{remark}\em
For an operator $E^\Sigma$ satisfying the second part of
Proposition~\ref{prp:propeq}, we call
\begin{equation}\label{eq:p-initial-data}
	E^\Sigma[\phi_0, \ldots, \phi_{l-1}]=0
\end{equation}
an \emph{initial data characterization} for the geometric 
equation $E[\phi]=0$ or in short
a set of \emph{$E$-initial data conditions} or a \emph{$E$-initial data
system}. Clearly, the operator $E^\Sigma$ is not uniquely fixed. For
instance, its components may contain many redundant equations. Thus, in
practice, once some $E$-initial data conditions have been obtained, they
will be significantly simplified by eliminating as many higher order (in
spatial derivatives) terms as possible. Also, when some of the
components of $E^\Sigma[\phi_0,\ldots,\phi_{l-1}]=0$ can be used to
directly solve for one of the arguments, say $\phi_{l-1}$, in terms of
the remaining ones, we can split the initial data system into
(a) $\phi_{l-1}=\cdots$ and (b) a system involving only the remaining arguments,
$E^{\prime\Sigma}[\phi_0,\ldots,\phi_{l-2}]=0$. When presenting an
$E$-initial data system, we will omit from $E^\Sigma$ those components that
can be rewritten as type (a) and only write the remaining components of
type (b), reduced to the smallest convenient set of arguments. Of
course, the derivation of the initial data system will provide the
information about how all type (a) components can be recovered.
\end{remark}
\begin{remark}\label{rem:isometric-embedding}\em
The geometric meaning of Proposition \ref{prp:propeq} is that the 
equations that characterize the initial data of the equation $E[\phi]=0$
given by \eqref{eq:p-initial-data} yields a set of necessary and sufficient 
conditions for the existence of an isometric embedding between a Riemannian
manifold $(\Sigma,h)$, that is isometric to the Cauchy hypersurface
$\Sigma$, and a globally hyperbolic $\Lambda$-vacuum solution $(M,g)$ of
the Einstein equations possessing a corresponding solution of
$E[\phi]=0$, with $\Sigma$ as its Cauchy hypersurface. 
Recall that any (vacuum) Cauchy data always 
has a {\em maximal} globally hyperbolic 
extension \cite{GEROCH-BRUHAT} that does not necessarily agree with the maximal 
analytic extension of a vacuum spacetime (the latter might not even be 
globally hyperbolic). Therefore the globally
hyperbolic spacetime $(M,g)$ considered in proposition \ref{prp:propeq}
shall be understood as the maximal globally hyperbolic extension of 
the Cauchy data. In this sense we may generalize the assumptions of 
Proposition \ref{prp:propeq} and speak of an initial data characterization
of the geometric equation $E[\phi]=0$ in a given ambient vacuum 
$(M,g)$, by restricting to a globally hyperbolic development of any
partial Cauchy surface, even if $(M,g)$ is not globally hyperbolic.
\end{remark}

For our purposes, in order to establish that any particular propagation
equation has a well-posed initial value problem, it will be sufficient
to check that it belongs to the class that we called \emph{generalized
normally hyperbolic} in~\cite[Sec.2]{gpkh-ckid}. We showed how equations
from that class inherit the well-posedness properties from the better
known \emph{normally hyperbolic} class~\cite{bgp}. An operator $Q$ (of
order $l$) is generalized normally hyperbolic when it is
\emph{determined} (acts between vector bundles of equal rank) and there
exists an operator $Q'$ (of order $2m-l$, $m\ge 1$) such that
\begin{equation} \label{eq:adjugate}
	N[\phi] := Q'[Q[\phi]] = \square^m \phi + \text{l.o.t} ,
\end{equation}
where l.o.t stands for term of differential order lower than $2m$. That
is, the principal symbol of $N[\phi]$ is a power of the wave operator.
We call $Q'$ an \emph{adjugate operator} for $Q$.

One of the consequences~\cite[Lem.3]{gpkh-ckid} of generalized normal
hyperbolicity of an operator $Q$ is the non-degeneracy of its principal
symbol $\sigma_p(Q)$ as a numerical matrix for any $p\in T^*M$ that is
not null (recall that $\sigma_p(Q)$ is a vector bundle morphism obtained
by collecting the highest order terms of $Q$ and replacing $\nabla_a$
with multiplication by a covector $p_a$, and hence, for given frames on
the source and target vector bundles, the principal symbol is a
numerical matrix valued function of $p_a$). Therefore, to show that some
$Q$ cannot be generalized normally hyperbolic, it is sufficient to
exhibit at least one non-null ($p_a p^a \ne 0$) value of $p\in T^*M$ for
which $\sigma_p(Q)$ is singular (equivalently, it possesses at least one
left or right null-vector). To save the trouble of explicitly writing
down such a covector $p_a$, for instance when there is no single
canonical choice, the following is a useful result:

\begin{lemma} \label{lem:not-normhyp}
Suppose that $Q$ is a determined operator of order $l$ and there exists
a non-vanishing differential operator $Q'$ of order $l'$ such that
$Q'\circ Q = 0 + \text{l.o.t}$ (or $Q\circ Q' = 0 + \text{l.o.t}$),
where $0$ is to be interpreted as a special case of a differential
operator of order $l+l'$. Then $Q$ cannot be generalized
normally hyperbolic.
\end{lemma}

\begin{proof}
The hypotheses basically mean that $\sigma_p(Q') \sigma_p(Q) = 0$
(or $\sigma_p(Q) \sigma_p(Q') = 0$) with $\sigma_p(Q')$ not being
identically zero at least for some $x\in M$. Since $\sigma_p(Q')$
depends polynomially on $p\in T^*_x M$, it has rank $\ge 1$ on an open
dense subset of $T^*_x M$. Pick any non-null covector $p_a$ from that
set, so that the row (column) space of $\sigma_p(Q')$ has at least one
non-vanishing element, which is then a left (right) null-vector of
$\sigma_p(Q)$. Since this implies that $\sigma_p(Q)$ is singular for a
non-null covector $p_a$, the operator $Q$ cannot be generalized
normally hyperbolic.
\end{proof}

\subsection{Example: Killing initial data in an Einstein space} \label{sec:kid}
To show a practical application of Proposition~\ref{prp:propeq}
we review here the case of the 
\emph{Killing equation} as it was presented in~\cite[Subsect.2.1]{gpkh-ckid}
,
\begin{dgroup*}
\begin{dmath} \label{eq:k}
	\K_{ab}[v] \hiderel{=} \nabla_a v_b + \nabla_b v_a = 0
	\condition[]{$(E[\phi] = 0)$} ,
\end{dmath}
\intertext{whose solution vector fields $v^a$ are infinitesimal
isometries of the background metric. If we assume that the Einstein
tensor $G_{ab}=0$ then the corresponding propagation equations are}
\begin{dmath} \label{eq:k-v-propeq}
	\square v_a + \frac{2\Lambda}{n-2} v_a = 0
	\condition[]{$(Q[\phi] = 0)$} ,
\end{dmath}
\begin{dmath} \label{eq:k-h-propeq}
	\square h_{ab} - 2R^c{}_{ab}{}^d h_{cd} = 0
	\condition[]{$(P[\psi] = 0)$} ,
\end{dmath}
\end{dgroup*}
where symmetric tensors $h_{ab} = K_{ab}[v]$ ($\psi = E[\phi]$) coincide
with the target of the Killing equation. The propagation equation
$P[\psi] = 0$ happens to coincide with the harmonic gauge linearized
Einstein evolution equation, with $h_{ab}$ considered as a linearized
perturbation of the background metric. The propagation
identities~\eqref{eq:propeq} then take the form 
\begin{dgroup*}
\begin{dmath} \label{eq:k-propeq}
	\square \K_{ab}[v] - 2R^c{}_{ab}{}^d \K_{cd}[v]
		= \K_{ab}\left[\square v+\frac{2\Lambda}{n-2} v\right] \\
	{}\hspace{7em} (P[E[\phi]] \hiderel{=} \sigma[Q[\phi]] + \tau[G]) ,
\end{dmath}
\begin{dmath} \label{eq:k-propeq-conv}
	\square v_a + \frac{2\Lambda}{n-2}v_a
		= \nabla^b \K_{ab}[v] - \frac{1}{2} \nabla_a \K^b{}_b[v] \\
	{}\hspace{7em} (Q[\phi] \hiderel{=} \rho[E[\phi]]) .
\end{dmath}
\end{dgroup*}

To obtain the $\K$-initial data conditions, or more commonly the
\emph{Killing initial data (KID)} conditions, we must first introduce a
space-time split around a Cauchy surface $\Sigma \subset M$, $\dim M =
n$ and $\dim \Sigma = n-1$. Let us use Gaussian normal coordinates to
set up a codimension-$1$ foliation on an open neighborhood $U \supset
\Sigma$ by level sets of a smooth temporal function $t\colon U\to
\mathbb{R}$, of which $\Sigma = \{ t = 0 \}$ is the zero level set.
Choose $t$ such that $n_a = \nabla_a t$ is a unit normal to the level
sets of $t$. Let us identify tensors on $\Sigma$ by upper case Latin
indices $A, B, C, \ldots$, denote the pullback of the ambient metric to
$\Sigma$ by $g_{AB}$ and its inverse by $g^{AB}$, and also denote by
$h^a_A$ the injection $T_\Sigma\to TM$ induced by the foliation. Raising
and lowering the respective indices on $h^a_A$ with $g_{ab}$ and
$g_{AB}$, we get the corresponding injections and orthogonal projections
between $T\Sigma$, $T^*\Sigma$, $TM$ and $T^*M$. In our notation, all
covariant and contravariant tensors split according to
\begin{equation}
	v_a = v_0 n_a + h_a^A v_A, \quad
	u^b = -u^0 n^b + h^b_B u^B ,
\end{equation}
which we also denote by
\begin{equation}
	v_a \to \begin{bmatrix} v_0 \\ v_A \end{bmatrix} , \quad
	u^b \to \begin{bmatrix} u^0 \\ u^B \end{bmatrix} .
\end{equation}
Thus, in our convention, the ambient metric splits as
\begin{equation}
	g_{a b} \to \begin{bmatrix} -1 & 0 \\ 0 & g_{A B} \end{bmatrix} .
\end{equation}
Let $D_A$ denote the Levi-Civita connection on $(\Sigma,g_{AB})$,
depending on the foliation time $t$ of course, and let $\del_t =
\mathcal{L}_{-n}$ denote the Lie derivative with respect to the
future-pointing normal vector $-n^a$. The action of $\del_t$ extends to
$t$-dependent tensors on $\Sigma$ in the natural way. The
($t$-dependent) extrinsic curvature on $\Sigma$ is then defined by
\begin{equation}
	\pi_{A B} = \frac{1}{2} \del_t g_{A B}
\end{equation}
and the ambient spacetime connection decomposes as
\begin{equation}
	\nabla_a v_b
	\to \begin{bmatrix} \nabla_0 v_b \\ \nabla_A v_b \end{bmatrix} ,
\end{equation}
where 
\begin{align}
	\nabla_0 v_a \to
		\begin{bmatrix}
			\nabla_0 v_0 \\ \nabla_0 v_A
		\end{bmatrix}
	&= \begin{bmatrix}
			\del_t & 0 \\
			0 & \del_t \delta_A^B - \pi_{A}{}^{B}
		\end{bmatrix}
		\begin{bmatrix}
			v_0 \\ v_B
		\end{bmatrix} ,
	\\
	\nabla_A v_b \to
		\begin{bmatrix}
			\nabla_A v_0 \\ \nabla_A v_B
		\end{bmatrix}
	&= \begin{bmatrix}
			D_A & -\pi_A{}^C \\
			-\pi_{A B} & D_A \delta_B^C
		\end{bmatrix}
		\begin{bmatrix}
			v_0 \\ v_C
		\end{bmatrix} .
\end{align}
The ambient $\Lambda$-vacuum Einstein equations $R_{ab} - \frac{2\Lambda}{n-2}
g_{ab} = 0$ decompose as
\begin{equation} \label{eq:ricci-split}
	%R_{ab} - \frac{2\Lambda}{n-2} g_{ab} \to
	\begin{bmatrix}
		-\nabla_0 \pi - \pi \cdot \pi+\frac{2\Lambda}{n-2}
			& D^{C} \pi_{C B} - D_{B} \pi \\
		D^{C} \pi_{C A} - D_{A} \pi
			& \nabla_0 \pi_{AB} + \pi \pi_{AB} + r_{AB}-\frac{2\Lambda}{n-2}g_{AB}
	\end{bmatrix} = 0 ,
\end{equation}
where now $r_{AB}$ is the Ricci tensor of $g_{AB}$ on $\Sigma$, $\pi =
\pi_C{}^C$, $(\pi\cdot \pi)_{AB} = \pi_A{}^C \pi_{C B}$ and $\pi\cdot
\pi = (\pi\cdot \pi)_C{}^C$. Note that we have found it convenient to
use the $\nabla_0$ operator instead of $\del_t$, because of its
preservation of both the orthogonal splitting with respect to the
foliation and of the spatial metric, $\nabla_0 g_{AB} = \nabla_0 g^{AB}
= 0$. For convenience, we note the commutator 
\begin{equation}\label{eq:commutator}
	(\nabla_0 D_A - D_A \nabla_0) \begin{bmatrix} v_0 \\ v_B \end{bmatrix}
	= -\pi_A{}^C D_C \begin{bmatrix} v_0 \\ v_B \end{bmatrix}
	+ \begin{bmatrix}
			0 \\ (D^C\pi_{AB} - D_B\pi_A{}^C)
		\end{bmatrix} v_C \;,
\end{equation}
and the identity
\begin{dmath} \label{eq:driem}
	\nabla_0 r_{ABCD}
	= -D_{(A} D_{C)} \pi_{BD} + D_{(A} D_{D)} \pi_{BC}
			+ D_{(B} D_{C)} \pi_{AD} - D_{(B} D_{D)} \pi_{AC} \\
		- \left( \pi_{[A|}{}^E r_{E|B]CD}
			+ \pi_{[C|}{}^E r_{ABE|D]} \right) \;,
\end{dmath}
which can be obtained by splitting the Bianchi identity
$\nabla_{[a} R_{bc]de} = 0$.
Taking the trace and using the vacuum equations yields
\begin{dmath}\label{eq:dricc}
	\nabla_0 r_{AC}
	= D_A D_C \pi - D^B D_B \pi_{AC} - 2 r_{ABCD} \pi^{BD} .
\end{dmath}

According to Proposition~\ref{prp:propeq} and the specific
identity~\eqref{eq:k-propeq}, the Killing equation $\K_{ab}[v] = 0$ is
satisfied when $v_a$ is any solution of~\eqref{eq:k-v-propeq} where both
$\left.\K_{ab}[v]\right|_\Sigma = \left.\nabla_0\K_{ab}[v]\right|_\Sigma = 0$.
Using respectively $\K_{00}[v]=0$ and $\K_{0B}[v]=0$ to eliminate the
time derivatives of $v_0$ and $v_B$ from these conditions, while also
eliminating the time derivatives of $\pi_{AB}$ using the
$\Lambda$-vacuum Einstein equations~\eqref{eq:ricci-split}, we obtain
the well-known Killing initial data (KID)
conditions~\cite{beig-chrusciel} in the presence of a cosmological
constant:
\begin{dgroup}[indentstep=6em] \label{eq:kid}
\begin{dmath} \label{eq:kid0}
	D_A v_B + D_B v_A - 2\pi_{AB} v_0 = 0 ,
\end{dmath}
\begin{dmath} \label{eq:kid1}
	D_{A} D_{B} v_{0}
	+ (2(\pi\cdot\pi)_{A B} - \pi \pi_{AB} - r_{AB}) v_0 \\
	- 2 \pi_{(B}{}^{C} D_{A)} v_C
	- (D^{C}{\pi_{A B}}) v_{C} + \frac{4\Lambda}{n-2}g_{AB} v_0 = 0 .
\end{dmath}
\end{dgroup}

\section{Closed Conformal Killing-Yano initial data} \label{sec:ccyk}

Consider an $n$-dimensional Lorentzian manifold $(M,g)$ satisfying the
vacuum Einstein equations with a cosmological constant $\Lambda$,
$R_{ab} - \frac{1}{2} R g_{ab} + \Lambda g_{ab}=0$ or equivalently
$R_{ab} = \frac{2\Lambda}{n-2} g_{ab}$. We will restrict ourselves to
dimensions $n>2$. Let $Y_{ab} = Y_{[ab]}$ be a $2$-form. The
\emph{conformal Killing-Yano (CYK)\footnote{We use the non-standard abbreviation 
CYK instead of CKY in order to use more consonant
abbreviations cCYKID, CYKID for (closed) Killing-Yano initial data rather than the alternative
CKYID, cCKYID that result from CKY.}} equation in $n$~dimensions is
\begin{dmath} \label{eq:cky-def}
	\CKY_{a\colon bc}[Y] \coloneq
	2\nabla_{a}{Y_{b c}}
	-\nabla_{b}{Y_{c a}}
	+\nabla_{c}{Y_{b a}}
	\\
	+\frac{3}{n-1} g_{ab} \nabla^d Y_{c d}
	-\frac{3}{n-1} g_{ca} \nabla^d Y_{b d}
	\hiderel{=} 0 .
\end{dmath}
In our index notation $a{:}bc$, the $:$ only serves to visually separate
groups of indices. When both $\CKY[Y]=0$ and the exterior derivative $\d
Y = 0$, we call $Y_{ab}$ a \emph{closed conformal Killing-Yano (cCYK)}
$2$-form.

\begin{remark}\em
The $\CKY_{a:bc}$ operator takes values in $3$-tensors that transform
pointwise irreducibly under $SO(1,n-1)$ (the group of orientation
preserving linear transformations respecting the Lorentzian metric
$g_{ab}$), which is traditionally labelled by the \emph{Young tableau}
{\LARGE\ytableaushort{{\scalebox{.45}{\raisebox{.75ex}{$b$}}}{\scalebox{.45}{\raisebox{.75ex}{$a$}}},{\scalebox{.45}{\raisebox{.75ex}{$c$}}}}}.
Given any such tableau (consisting of left-aligned rows of boxes of
non-increasing length) filled with tensor indices, the corresponding
subspace carrying the irreducible representation is obtained by first
symmetrizing over the rows, then antisymmetrizing over the columns and
finally subtracting all the traces. The resulting representation is
always irreducible, with the possible exception of tableaux with columns
of length exactly $n/2$ (due to the possibility of decomposing
$(n/2)$-forms into self-dual and anti-self-dual subspaces). But such
exceptions only occur for some dimensions and signatures and we will not
encounter them below (since we use Lorentzian signature and real
representation), with the exception of 3-forms $\yd{1,1,1}$ in dimension
$n=6$. If we do not subtract the traces, then we obtain a
subspace transforming irreducibly under $GL(n)$ (the group of general
linear transformations), but which may be reducible with respect to
$SO(1,n-1)$. Below, we will freely use \emph{Young diagrams} (unfilled
Young tableaux) to label other irreducible tensor representations.
Although we will not need more of them, basic facts about Young diagrams
and their relation to $GL(n)$ and $SO(p,q)$ representation theory can be
found in~\cite{fulton, fulton-harris, hamermesh}.
\end{remark}

The representation type of the conformal Killing-Yano operator implies
that it must satisfy the following identities
\begin{equation}
	\CKY_{a:(bc)}[Y] = \CKY_{[a:bc]}[Y] = g^{ab} \CKY_{a:bc}[Y] = 0 ,
\end{equation}
which can be straightforwardly verified from its definition
in~\eqref{eq:cky-def}. Our use of $:$ is to separate out the
antisymmetric index group $bc$.

The covariant derivative of $Y_{ab}$ decomposes as
\begin{dgroup*}
\begin{dmath} \label{eq:dY-decomp}
	(\yd{1}\, \nabla_a) ( \yd{1,1}\, Y_{bc} )
	= \frac{1}{3} \yd{2,1}\, \CKY_{a:bc}[Y]
		+ \frac{1}{6} (\yd{1,1,1}\, \d Y)_{abc}
		- \frac{2}{n-1} g_{a [b} (\yd{1}\, \delta Y)_{c]} ,
\end{dmath}
\begin{dmath} \label{eq:dR-def}
	\text{with} \quad
	(\d Y)_{abc} \coloneq 2 (\nabla_a Y_{bc} + \nabla_b Y_{ca} + \nabla_c
	Y_{ab}) ,
\end{dmath}
\begin{dmath} \label{eq:div-def}
	\text{and} \quad
	(\delta Y)_b \coloneq \nabla^a Y_{ab} ,
\end{dmath}
\end{dgroup*}
where we have prefixed $Y_{ab}$ and various operators with Young
diagrams indicating that they take values in the corresponding
irreducible $SO(1,n-1)$ representation. The most important information
contained in~\eqref{eq:dY-decomp} can be summarized
representation-theoretically by the decomposition of the following
tensor product of representations into irreducible ones: $\yd{1} \:
\yd{1,1} = \yd{1} + \yd{2,1} + \yd{1,1,1}$. Since each irreducible
representation on the right-hand side appears with multiplicity one,
Schur's lemma guarantees that the projection of $\nabla_a Y_{bc}$ onto
the corresponding representation is uniquely fixed up to a scalar
multiple, which we explicitly fix in the definitions~\eqref{eq:cky-def},
\eqref{eq:dR-def} and~\eqref{eq:div-def}. In the case of multiplicity
greater than one, there will exist multiple independent projectors onto
the same representation, with the number of independent ones equal to
the multiplicity. Of course, the choice of basis in this space of
projectors is not unique and has to be made by hand.

The above basic ideas from representation theory will help us carry out
an exhaustive search for a \emph{covariant}, \emph{second order}
propagation identity of the form~\eqref{eq:propeq} for the cCYK
equations. A priori, the order of the differential operators
in~\eqref{eq:propeq} is not bounded, nor do they have to be covariant,
but for practical reasons, we have restricted our search to operators
$P$ and $Q$ that are of second order are covariantly constructed from the
Levi-Civita connection $\nabla_a$, the metric $g_{ab}$ and the Riemann tensor%
	\footnote{Our conventions are $2\nabla_{[a}\nabla_{b]} v_c =
	R_{abc}{}^d v_d$ and $R_{ab} = R_{acb}{}^c$.} %
$R_{abcd}$. Expanding the search to higher orders would be prohibitively
expensive, at least without significant automation of our methods. In
any case, our search will succeed (Theorem~\ref{thm:ccyk-propeq}) in all
dimensions ($n>2$) except $n=4$. The $4$-dimensional case will be
handled separately in Section~\ref{sec:cykid}.

Let us also briefly remark that the two separate $\CKY[Y] = 0$ and $\d Y
= 0$ equations can be combined into the single equivalent equation
\begin{equation} \label{eq:ccyk-combined}
	\nabla_a Y_{bc} - \frac{2}{n-1} g_{a[b} \nabla^d Y_{c]d} = 0 ,
\end{equation}
where the left-hand side is just the traceless part of $\nabla_a
Y_{bc}$. The tensor type of this equation is not irreducible in the
sense of $SO(1,n-1)$ representations. So it would not be as helpful in
the representation-theoretic exhaustive search described above.
Projecting this equation onto the $\yd{2,1}$ and $\yd{1,1,1}$ tensor
types recovers the original separate equations, also demonstrating the
complete equivalence of the two formulations.

\subsection{Dimensions $n<4$}

The lowest dimension in which the CYK operator makes sense is $n=2$.
However, in that case, the $SO(1,1)$ representations $\yd{2,1}$ and
$\yd{1,1,1}$ are both 0-dimensional, meaning that the equations $\CYK[Y]
= 0$ and $\d Y = 0$ are both trivial conditions of the form $0=0$.

In dimension $n=3$, we can represent any 2-form as $Y_{ab} =
\eta_{ab}{}^c Y_c$, where $Y_c$ is a $1$-form and $\eta_{abc}$ is the
Levi-Civita tensor. We then have the identities
\begin{gather}
	-\frac{1}{3} \eta_{(a}{}^{cd} \CKY_{b):cd}[Y]
	= \nabla_a Y_b + \nabla_b Y_a - \frac{2}{3} g_{ab} \nabla^c Y_c ,
	\\
	-\frac{1}{2} \eta^{abc} (\d Y)_{abc} = \nabla^c Y_c .
\end{gather}
This means that $\CKY[Y] = 0$ is equivalent to the \emph{conformal
Killing} equation on $Y_c$,
\begin{equation} \label{eq:ck}
	\CK_{ab}[Y] := \nabla_a Y_b + \nabla_b Y_a
		- \frac{2}{n} g_{ab} \nabla^c Y_c = 0,
\end{equation}
while imposing the additional condition $(\d Y)_{abc} = 0$, or the
equivalent $\nabla^c Y_c = 0$, turns it into the Killing equation. The
propagation identities and initial data for the conformal Killing
equation were found in our previous work~\cite{gpkh-ckid}, where we also
reviewed the analogous well-known results for the Killing equation as
well as their history.

Thus, in the rest of this work we concentrate on dimension $n\ge 4$.

\subsection{Propagation identity in dimension $n>4$} \label{sec:ccyk-gendim}

Our search strategy has the following steps: (a) identify a basis for
the potential $P$, $Q$, $\rho$, and $\sigma$ operators in the
propagation identity~\eqref{eq:propeq}, (b) find the most general
solution for these operators and identify the free parameters that it
depends on, (c) check for which values of the free parameters the
operators $P$ and $Q$ are generalized normally hyperbolic. The result of
this search, recorded in Theorem~\ref{thm:ccyk-propeq}, is that there
exists a $5$-parameter family of identities satisfying all of our search
criteria.

\paragraph{(a)}
We start by listing the basis elements of for all the operators we want
to parametrize. The following schematic identity illustrates the number
of basis elements and their labels:

\begin{dgroup} \label{eq:gen-ccyk}
\begin{dmath}[indentstep=9.5em] \label{eq:gen-ccyk-P}
	\begin{matrix}
		\yd{2,1} \\ \yd{1,1,1}
	\end{matrix}
	\begin{bmatrix}
		P^{1,2,3,4,5,6} & P^{7,8} \\
		\hat{P}^{5,6} & \hat{P}^{1,2,3,4}
	\end{bmatrix}
	\begin{matrix}
		\yd{2,1} \\ \yd{1,1,1}
	\end{matrix}
	\begin{bmatrix}
		\CKY \\ \d
	\end{bmatrix}
	\yd{1,1}
	-
	\begin{matrix}
		\yd{2,1} \\ \yd{1,1,1}
	\end{matrix}
	\begin{bmatrix}
		\sigma^1 \\ \hat{\sigma}^1
	\end{bmatrix}
	\yd{1,1}
	\begin{bmatrix}
		\rho^1 & \rho^2
	\end{bmatrix}
	\begin{matrix}
		\yd{2,1} \\ \yd{1,1,1}
	\end{matrix}
	\begin{bmatrix}
		\CKY \\ \d
	\end{bmatrix}
	\yd{1,1}
	=
	\begin{matrix}
		\yd{2,1} \\ \yd{1,1,1}
	\end{matrix}
	\begin{bmatrix}
		T^{1,2,3,4,5,6,7,8} \\ \hat{T}^{1,2,3,4}
	\end{bmatrix}
	\yd{1,1} \:,
\end{dmath}
\begin{dmath} \label{eq:gen-ccyk-Q}
	\yd{1,1}
	\begin{bmatrix}
		\rho^1 & \rho^2
	\end{bmatrix}
	\begin{matrix}
		\yd{2,1} \\ \yd{1,1,1}
	\end{matrix}
	\begin{bmatrix}
		\CKY \\
		\d
	\end{bmatrix}
	\yd{1,1}
	= \yd{1,1} \: Q^{1,2,3,4} \: \yd{1,1} \:.
\end{dmath}
\end{dgroup}
To visually help the reader, we have inserted Young diagram labels to
illustrate the irreducible tensor representations that each operator
acts between.

To explain the size of each basis, we first need to define precisely what
we mean by a \emph{second order operator}. Obviously, it cannot contain
terms with more than two iterated $\nabla_a$ derivatives. But each
subleading term should also be of \emph{total order} two. Being
covariant, each such subleading term consists of a number of iterated
$\nabla_a$ derivatives multiplied by copies of the Riemann tensor $R$
and its covariant derivatives. We count the total order as follows: it
is additive for products, $\nabla^k$ has total order $k$, $\nabla^l R$
has total $2+l$, the cosmological constant $\Lambda$ has total order
$2$, while other constants, the metric $g_{ab}$ and tensor
contractions have total order zero. It is easy to see that this total
order is preserved by the Leibniz rule, exchange of covariant
derivatives and substitution of Einstein's equations, while it is
additive under operator composition.

Thus, in the identities~\eqref{eq:gen-ccyk}, $\CYK$, $\d$, $\rho^1$,
$\rho^2$, $\sigma^1$, $\hat{\sigma}^1$ are all of total order one,
$P^i$, $\hat{P}^i$, $Q^i$ are all of total order two, and $T^i$,
$\hat{T}^i$ are all of total order three. To see how many independent
ways there are to combine two $(\yd{1}\, \nabla_a)$ derivatives with a
tensor like $\yd{2,1}\, C_{a:bc}$, consider the tensor product
decomposition
\begin{equation}
	(\yd{1} \: \yd{1}) \: \yd{2,1}
	= (\mathbb{R} + \yd{2}) \: \yd{2,1}
	= (\yd{2,1})
		+\left(\yd{4,1}
		 +   \yd{3,2}
		 +   \yd{3,1,1}
		 +   \yd{2,2,1}
		 +   \yd{3}
		 +2\,\yd{2,1}
		 +   \yd{1,1,1}
		 +   \yd{1} \right) ,
\end{equation}
where in the first product, the trivial representation $\mathbb{R}$
corresponds to $\square= g^{ab}\nabla_a\nabla_b$, while $\yd{2}$
corresponds to the traceless symmetrized projection of
$\nabla_a\nabla_b$. We need only consider symmetrized derivatives
$(\yd{1}\:\yd{1})\: \nabla_{(a} \nabla_{b)}$, since the antisymmetric
part is equivalent to contractions with the Riemann tensor.
The total multiplicity of $\yd{2,1}$ appearing on
the right-hand side is $3=1+2$, thus there are three independent ways of
applying two derivatives to $C_{a:bc}$. This number automatically counts
all possible index permutations, index contractions and products with
metric $g_{ab}$ or Levi-Civita $\eta_{a_1\cdots a_n}$ tensors, as all
such operations are $SO(1,n-1)$ equivariant. Similarly, we can work out
the number of independent ways to combine the Riemann tensor $R_{abcd}$
with a tensor argument by recalling~\cite[p.193]{penrose-spinor} that
the $\Lambda$-vacuum Weyl tensor $W_{abcd} = R_{abcd} - \frac{4\Lambda
g_{a[c} g_{d]b}}{(n-1)(n-2)}$ (due to the algebraic Bianchi identity and
being fully traceless) belongs to a representation of type $\yd{2,2}$,
while the (appropriately symmetrized traceless) derivatives $\nabla W$,
$\nabla\nabla W$, \ldots (due to the differential Bianchi identity and
its contractions) have independent projections only onto the respective
representations $\yd{3,2}$, $\yd{4,2}$, \ldots, with other projections
being expressible in terms of lower order derivatives.

For sufficiently large $n$ (cf.~Remark~\ref{rmk:littlewood}), the
following tensor decomposition product tables, show how all the operator
basis elements from~\eqref{eq:gen-ccyk} fit into the above scheme, where
to save space we have dropped all summands that are irrelevant for
identity~\eqref{eq:gen-ccyk}, with multiplicity indicated by listing
multiple basis element labels. For economy of notation, we refer
directly to the Riemann tensor and its derivatives $R$ and $\nabla R$,
rather than the Weyl tensor expressions $W$ and $\nabla W$, since in
explicit computations $W_{abcd}$ must anyway be expressed in terms of
$R_{abcd}$ and $\Lambda$ (cf.~Remark~\ref{rmk:weyl}). For the second
order operators:
\begin{center}
\begin{tabular}{@{}CR|C|C|C@{}}
\toprule
                   &            & \yd{1,1}\,Y_{ab}                             & \yd{2,1}\,C_{a:bc}                    & \yd{1,1,1}\,\Xi_{abc}             \\[0.5ex] \hline &&&& \\[-2ex] %\midrule
\nabla             & \yd{1}     & \yd{2,1}\,\sigma^1+ \yd{1,1,1}\,\hat{\sigma}^1 & \yd{1,1}\,\rho^1                      & \yd{1,1}\,\rho^2                  \\[0.5ex] \hline &&&& \\[-2ex] %\midrule
\square            & \mathbb{R} & \yd{1,1}\,Q^1                                & \yd{2,1}\,P^1                         & \yd{1,1,1}\,\hat{P}^1             \\
\nabla\nabla       & \yd{2}     & \yd{1,1}\,Q^2
& \yd{2,1}\,P^{2,3}+\yd{1,1,1}\,\hat{P}^5 & \yd{2,1}\,P^7+\yd{1,1,1}\,\hat{P}^2 \\
R                  & \yd{2,2}   & \yd{1,1}\,Q^3                                & \yd{2,1}\,P^{4,5}+\yd{1,1,1}\,\hat{P}^6 & \yd{2,1}\,P^8+\yd{1,1,1}\,\hat{P}^3 \\
\Lambda            & \mathbb{R} & \yd{1,1}\,Q^4                                & \yd{2,1}\,P^{6}                         & \yd{1,1,1}\,\hat{P}^4
\\ \bottomrule
\end{tabular}
\end{center}
And the same for the third order operators:
\begin{center}
\begin{tabular}{@{}CR|C|CCC@{}}
\toprule
                   &            &                                            & \multicolumn{3}{C}{(\yd{1}\:\yd{1,1}) \nabla_a Y_{bc}}                          \\[1ex] \cline{4-6} & & & & & \\[-2ex]
                   &            & \yd{1,1}\,Y_{ab}                             & \yd{1}\,\delta Y & \yd{2,1}\,\CKY[Y]                     & \yd{1,1,1}\,\d Y                 \\[0.5ex] \hline &&&&& \\[-2ex] %\midrule
\square\nabla      & \yd{1}     & \yd{2,1}\,T^1+\yd{1,1,1}\,\hat{T}^1         &              &                                     &                                  \\
\nabla\nabla\nabla & \yd{3}     & \yd{2,1}\,T^2                                &              &                                     &                                  \\
\nabla R           & \yd{3,2}   & \yd{2,1}\,T^7                                &              &                                     &                                  \\
R                  & \yd{2,2}     &                                            & \yd{2,1}\,T^6  & \yd{2,1}\,T^{3,4}+\yd{1,1,1}\,\hat{T}^2 & \yd{2,1}\,T^5+\yd{1,1,1}\,\hat{T}^3  \\
\Lambda            & \mathbb{R}   &                                            &                & \yd{2,1}\,T^8                       & \yd{1,1,1}\,\hat{T}^4  \\
	\bottomrule
\end{tabular}
\end{center}

\begin{remark} \label{rmk:littlewood}\em
For sufficiently large $n$ (larger than the total number of boxes
involved in the product, for instance) these product tables can be
checked using Littlewood's rule~\cite[Thm.I]{littlewood-ortho}
(cf.~\cite{king-ortho, koike-terada-ortho} for complete proofs and
modern generalizations). In our case, we need to take $n\ge 8$ for
Littlewood's rule to apply for all products that we are interested in.
For small values of $n$, namely $4<n<8$, the above multiplication tables
can be checked using computer algebra~\cite{sagemath,lie-cas}.
We have found
exceptions only in dimension $n=6$. Basically, the 3-form representation
$\yd{1,1,1}$ becomes reducible for $n=6$ in Lorentzian signature, and
splits into the eigen-subspaces of the Hodge ${*}$ operator. Thus the
tensor product tables need to be rewritten taking that into account. In
addition, the multiplicity of $\yd{2,1}$ in the product $\yd{2,2}\;
\yd{2,1}$ is $3$ instead of $2$. Thus, we cannot claim that our lists of
operators give complete bases in dimension $n=6$.
\end{remark}

The simplest possibilities are for the $\sigma$ and $\rho$ operators:
\begin{align}
	\sigma^1_{a:bc}[Y] &= \CKY_{a:bc}[Y] , \\
	\hat{\sigma}^1_{abc}[Y] &= (\d Y)_{abc} ; \\
	\rho^1_{ab}[C] &= \nabla^c C_{c:ab}[Y] , \\
	\rho^2_{ab}[\Xi] &= (\delta \Xi)_{ab} \coloneq \nabla^c \Xi_{cab} .
\end{align}

The definitions of the operators $P^i$, $\hat P^i$, $T^i$, $\hat{T}^i$
and $Q^i$ are somewhat lengthy and their precise form can be found in
Appendix~\ref{app:tensors}. What is salient about these operators are
their composition rules, which are reported in the next paragraph.

\paragraph{(b)}
The left-hand side in the schematic identity~\eqref{eq:gen-ccyk-P} is
parametrized by the coefficients in front of the $P^i$, $\hat{P}^i$, and
$\rho^i$ terms (up to rescaling, there is a unique possibility for each
of the $\sigma^1$ and $\hat{\sigma}^1$ operators, so their coefficients
can be absorbed into those of the $P^i$ and $\hat{P}^i$, respectively),
while the right-hand side is parametrized by the coefficients of $T^i$
and $\hat{T}^i$. By explicitly computing all the relevant operator
compositions, we end up with the following matrix identities:
\begin{subequations}
\begin{equation} \scriptsize \label{eq:PT-matrix}
	\begin{bmatrix}[l]
		P^1 \circ \CKY \\
		P^2 \circ \CKY \\
		P^3 \circ \CKY \\
		P^4 \circ \CKY \\
		P^5 \circ \CKY \\
		P^6 \circ \CKY \\
			\cmidrule(lr){1-1}
		P^7 \circ \d \\
		P^8 \circ \d \\
			\cmidrule(lr){1-1}
		\sigma^1 \circ \rho^1 \circ \CKY \\
		\sigma^1 \circ \rho^2 \circ \d
	\end{bmatrix}
	= \begin{bmatrix}
		1 & 0 & 0 & 0 & 0 & 0 & 0 & 0 \\
		0 & 3\frac{(n-2)}{(n-1)} & 0 & 0 & 0 & -\frac{9}{2}(n-2) & 0 & 0 \\
		-2 & -3\frac{(n-4)}{(n-1)} & 0 & 2 & 3 & -\frac{3}{2}(n+8) & 3 & 0 \\
		0 & 0 & 1 & 0 & 0 & 0 & 0 & 0 \\
		0 & 0 & 0 & 1 & 0 & 0 & 0 & 0 \\
		0 & 0 & 0 & 0 & 0 & 0 & 0 & 1 \\
			\cmidrule(lr){1-8}
		2 & -6 & 1 & -2 & 0 & -3(n-1) & 6 & -\frac{12}{n-2} \\
		0 & 0 & 0 & 0 & 1 & 0 & 0 & 0 \\
			\cmidrule(lr){1-8}
		2 & 3\frac{(n-4)}{(n-1)} & 0 & -2 & -3 & \frac{3}{2}(n+8) & -3 & 0 \\
		2 & -6 & 1 & -2 & 0 & -3(n-1) & 6 & -\frac{12}{n-2}
	\end{bmatrix}
	\begin{bmatrix}
		T^1 \\
		T^2 \\
		T^3 \\
		T^4 \\
		T^5 \\
		T^6 \\
		T^7 \\
		T^8
	\end{bmatrix} ,
\end{equation}

\begin{equation} \label{eq:hatPT-matrix}
	\begin{bmatrix}[l]
		\hat{P}^1 \circ \d \\
		\hat{P}^2 \circ \d \\
		\hat{P}^3 \circ \d \\
		\hat{P}^4 \circ \d \\
			\cmidrule(lr){1-1}
		\hat{P}^5 \circ \CKY \\
		\hat{P}^6 \circ \CKY \\
			\cmidrule(lr){1-1}
		\hat{\sigma}^1 \circ \rho^1 \circ \CYK \\
		\hat{\sigma}^1 \circ \rho^2 \circ \d
	\end{bmatrix}
	= \begin{bmatrix}
		1 & 0 & 0 & 0 \\
		-2 & 0 & -2 & \frac{12}{n-2} \\
		0 & 0 & 1 & 0 \\
		0 & 0 & 0 & 1 \\
			\cmidrule(lr){1-4}
		1 & -1 & \frac{1}{2} & 0 \\
		0 & 1 & 0 & 0 \\
			\cmidrule(lr){1-4}
		2 & -2 & 1 & 0 \\
		2 & 0 & 2 & -\frac{12}{n-2}
	\end{bmatrix}
	\begin{bmatrix}
		\hat{T}^1 \\
		\hat{T}^2 \\
		\hat{T}^3 \\
		\hat{T}^4
	\end{bmatrix} .
\end{equation}
\end{subequations}

The goal now is to find a set of coefficients for the left-hand side
of the propagation identity~\eqref{eq:gen-ccyk-P} such that the
right-hand side is identically zero. It is easy to see that such
coefficients correspond exactly to left null-vectors of the matrix
$\mathcal{C}$ defined by the identity
\begin{equation}
	\begin{bmatrix}[l|l]
		P^{1,2,3,4,5,6} \circ \CKY & \\ \cmidrule(lr){1-1}
		P^{7,8} \circ \d & \\ \cmidrule(lr){1-1} \cmidrule(lr){2-2}
		& \hat{P}^{1,2,3,4} \circ \d \\ \cmidrule(lr){2-2}
		& \hat{P}^{5,6} \circ \CKY \\ \cmidrule(lr){1-1} \cmidrule(lr){2-2}
		\sigma^1\circ \rho^1 \circ \CKY & \hat{\sigma}^1 \circ
		\rho^1 \circ \CKY \\
		\sigma^1\circ \rho^2 \circ \d & \hat{\sigma}^1 \circ \rho^2 \circ \d
	\end{bmatrix}
	= \mathcal{C}
	\begin{bmatrix}[c|c]
		T^{1,2,3,4,5,6,7,8} & \\ \cmidrule(lr){1-1}  \cmidrule(lr){2-2} 
		& \hat{T}^{1,2,3,4}
	\end{bmatrix} .
\end{equation}
The matrix $\mathcal{C}$ can be easily constructed from the blocks of
the coefficient matrices in~\eqref{eq:PT-matrix}
and~\eqref{eq:hatPT-matrix}. It is now a matter of basic linear algebra
to check that (for $n>2$) $\mathcal{C}$ has a $5$-dimensional left
null-space. The general $5$-parameter solution for a left-null vector
recorded in Theorem~\ref{thm:ccyk-propeq}.

\paragraph{(c)}
It remains to check that the possible $P$ and $Q$ operators in the
propagation identity~\eqref{eq:gen-ccyk} are generalized normally
hyperbolic (according to the definition given in
Section~\ref{sec:propeq}). To that end, only the coefficients of the
operators that contribute to the principal symbol are important (those
actually of second differential order). For $P$, these are
$P^{1,2,3}$, $P^7$, $\hat{P}^{1,2}$ and $\hat{P}^5$, while for $Q$
these are $Q^{1,2}$. Note that, explicitly working out the operator
compositions in~\eqref{eq:gen-ccyk-Q} gives the following relation
between the coefficients of $\rho^{1,2}$ and those of $Q^{1,2,3,4}$:
\begin{equation} \label{eq:rhoQ-matrix}
	\begin{bmatrix}[l]
		\rho^1 \circ \CKY \\
		\rho^2 \circ \d
	\end{bmatrix}
	= \begin{bmatrix}
		2 & -\frac{(n-4)}{(n-1)} & -1 & \frac{4}{n-2} \\
		2 & 2 & 2 & -\frac{8}{n-2}
	\end{bmatrix}
	\begin{bmatrix}
		Q^1 \\
		Q^2 \\
		Q^3 \\
		Q^4
	\end{bmatrix} .
\end{equation}

Checking generalized normal-hyperbolicity of an operator comes down to
pa\-ra\-met\-ri\-zing an ansatz for the adjugate operator and checking whether
the key identity~\eqref{eq:adjugate} can be satisfied for some values of the
parameters. In Appendix~\ref{app:hyp-conds}, we have recorded the
necessary and sufficient conditions for generalized normal hyperbolicity
of $Q$ in Lemma~\ref{lem:hyp-Q-ccyk} and of $P$ in
Lemma~\ref{lem:hyp-P-ccyk}. Applied to the family of identities obtained
in step (b), we find that for $n>2$ (with the exception of $n=4$) a
generic element of the family both $P$ and $Q$ are generalized normally
hyperbolic, with the exceptional values of the parameters consisting of
the union of certain hyperplanes. The full result is recorded in
Theorem~\ref{thm:ccyk-propeq}.

\bigskip

We are now ready to state the main result of this section in
\begin{theorem} \label{thm:ccyk-propeq}
For $n>2$, there exists the following identity of the
form~\eqref{eq:gen-ccyk} (with vanishing right-hand side):
\begin{equation} \label{eq:ccyk-propeq}
	P = \begin{bmatrix}
		\displaystyle
		\sum_{i=1}^{6} p_i P^i &
			p_7 P^7 + p_8 P^8 \\
		\hat{p}_5 \hat{P}^5 + \hat{p}_6 \hat{P}^6 &
			\displaystyle
			\sum_{i=1}^{4} \hat{p}_i \hat{P}^i
	\end{bmatrix} ,
	\quad
	Q = r_1 \rho^1\circ \CYK + r_2 \rho^2 \circ \d ,
\end{equation}
with
\begin{equation}
\begin{aligned}
	p_1 &= x  , &
	\hat{p}_1 &= y  , \\
	p_2 &= -\frac{x}{n-2}  , &
	\hat{p}_2 &= \frac{x_1+y+z}{2} , \\
	p_3 &= \frac{x}{3} - y_1 , &
	\hat{p}_3 &= \frac{2y+z}{2} - y_1 , \\
	p_4 &= \frac{x}{6}  , &
	\hat{p}_4 &= -\frac{6(y+z)}{n-2} + \frac{12}{n-2} y_1 , \\
	p_5 &= -x  , &
	\hat{p}_5 &= z , \\
	p_6 &= -\frac{2x}{n-2} , &
	\hat{p}_6 &= z - 2 y_1 , \\
	p_7 &= -\frac{x}{6} - \frac{x_1}{2} - y_1 , &
	r_1 &= y_1 , \\
	p_8 &= -x  , &
	r_2 &= -\frac{x_1}{2} - y_1 ,
\end{aligned}
\end{equation}
where $x$, $y$, $z$, $x_1$ and $y_1$ are free parameters. Moreover, for
$n>2$, necessary and sufficient conditions on these free parameters for
the generalized normal hyperbolicity of $Q$ consist of
\begin{subequations} \label{eq:ccyk-propeq-hyp}
\begin{equation}\label{eq:hyp-par1}
	x_1 \ne 0 , \quad
	y_1 \ne 0 ,
\end{equation}
and for $P$ they consist of
\begin{equation}\label{eq:hyp-par2}
\begin{gathered}
	x \ne 0, \quad
	(n-4) x \ne 0 , \quad
	y \ne 0 , \\
	x_1 (3z-x - 6y_1) \ne 0 ,
	\quad \text{and} \quad
	y_1 \ne 0 .
\end{gathered}
\end{equation}
\end{subequations}
For $n>4$ and $n\ne 6$, the family of operators $P$ and $Q$ in
\eqref{eq:ccyk-propeq} is the most general one of its kind.
\end{theorem}

\begin{proof}
The theorem follows from the calculations discussed in steps (a), (b)
and (c) above. As explained in Remark~\ref{rmk:littlewood}, the
structure of the tensor product decomposition tables from step (a)
allows us to claim that we have carried out an exhaustive search in all
dimensions $n>4$, with the exception of $n=6$.
\end{proof}

\begin{remark}\em
The propagation identity from the above theorem can now be used to
construct cCYKID conditions, for generic values of the free parameters.
The free parameters can also be chosen to simplify $P$ in some ways. For
instance, we can make its principal symbol block diagonal ($p_7 =
\hat{p}_5 = 0$) by setting $z=0$ and $x_1 = -(x+6y_1)/3$. But we
cannot make $P$ either lower ($p_7 = p_8 = 0$) or upper ($\hat{p}_5
= \hat{p}_6 = 0$) block triangular without violating at least one
of the hyperbolicity inequalities. As a consequence, it is also
impossible to decouple the $\CKY[Y] = 0$ equation from the $\d Y = 0$
equation ($p_7 = p_8 = r_2 = 0$).
\end{remark}

The exhaustive nature of the search which produced
Theorem~\ref{thm:ccyk-propeq} then leads to the following
\begin{corollary} \label{cor:cyk-propeq-nogo}
In dimensions $n>4$, $n\ne 6$, there does not exist a propagation
identity for the equation $\CKY[Y] = 0$ with operators $P$ and $Q$ of
total order 2.
\end{corollary}
Likely, dimension $n=6$ is not an exception to the corollary, but our
analysis would have to be extended to arrive at that conclusion
rigorously (cf.~Remark~\ref{rmk:littlewood}).

\subsection{Construction of cCYKID in dimension $n>4$} \label{sec:ccykid}

Let us denote by $C_{a:bc}$ the left-hand side
of~\eqref{eq:ccyk-combined}, the combined form of the cCYK operator.
Then we have the following integrability condition
\begin{multline} \label{eq:ccyk-ic}
	\nabla_{[d} C_{a]:bc}
		- \frac{g_{c[d} \nabla^e C_{a]:b}{}_e - g_{b[d} \nabla^e C_{a]:c}{}_e}{(n-2)} 
	\\
	=
	R_{da e[b} Y_{c]}{}^e
	+ 2 \Lambda \frac{g_{c[d} Y_{a]b} - g_{b[d} Y_{a]c}}{(n-2)^2}
	+ \frac{\left(g_{b[d} R_{a]cef}-g_{c[d} R_{a]bef}\right) Y^{ef}}{2 (n-2)}
	\\
	\eqcolon I_{da:bc}[Y] .
\end{multline}
The zeroth order operator $I_{da:bc}[Y]$ acting on $Y_{ab}$ is traceless
and antisymmetric in both groups of $:$ separated indices, but has no
other symmetries. It will be useful in giving the precise form of the
cCYKID conditions below.

\begin{theorem} \label{thm:ccykid}
Consider a globally hyperbolic Einstein $\Lambda$-vacuum Lorentzian ma\-ni\-fold,
$(M,g)$ of dimension $n>4$ with $R_{ab}=\frac{2\Lambda}{n-2}g_{ab}$, and
a Cauchy surface $\Sigma\subset M$. The necessary and sufficient
conditions yielding a set of \emph{closed conformal Killing-Yano initial
data} (cCYKID) for $Y_{ab}$ on $\Sigma$ are given by the following
equations, where we also indicate the provenance of each equation, and
each equality holds modulo the preceding ones.

\begin{subequations} \label{eq:ccykid}
\noindent
$\left.\frac{1}{3}(\CKY_{A:B0}[Y] + \frac{1}{2}(\d Y)_{AB0})\right|_\Sigma=0 \colon$
\begin{dmath}\label{eq:tfcdelectric}
	D_A Y_{B0} - \frac{1}{n-1} g_{AB} D^C Y_{C0} - \pi_A{}^C Y_{BC} = 0 \;,
\end{dmath}
$\left.\frac{1}{3}(\CKY_{A:BC}[Y] + \frac{1}{2}(\d Y)_{ABC})\right|_\Sigma=0 \colon$
\begin{dmath}\label{eq:tfcdmagnetic}
	D_A Y_{BC} - \frac{2}{n-2} g_{A[B|} D^D Y_{D|C]} \\
		+ 2\pi_{A[B} Y_{C]0}
		- \frac{2}{n-2} g_{A[B}
			(\pi g_{C]D} - \pi_{C]D}) Y^{D}{}_0 = 0 \;,
\end{dmath}
$\frac{(n-3)}{3(n-2)} \left.\nabla_0 \CKY_{(A:B)0}[Y]\right|_\Sigma=0 \colon$
\begin{dmath}\label{eq:CD0CYKIDAB0}
	2 I_{0(A\colon B)0}[Y] \hiderel{=} 
	(D_{C}\pi_{AB} - D_{(A}\pi_{B)C}) Y^{C}{}_{0}
	\\
	+ \pi \pi_{(A}{}^{C} Y_{B)C} 
	+ \pi_{(A|C} \pi^{CD} Y_{D|B)}
	+ r_{(A}{}^{C} Y_{B)C} 
	= 0 \;,
\end{dmath}
$-\frac{1}{6} \left.\nabla_0 \CKY_{A:BC}[Y]\right|_\Sigma=0 \colon$
\begin{dmath}\label{eq:CD0CYKIDABC}
	-I_{0A\colon BC}[Y] \hiderel{=}
	(D_{[C}{\pi_{B]}{}^{E}}) Y_{A E}
	+(D_{[C|}{\pi_{A}{}^{E}}) Y_{|B] E}
	+(D^{E}{\pi_{A [B}}) Y_{C] E}
	+\frac{1}{2}\left( r_{B C A}{}^{E} + 2\pi_{A [B} \pi_{C]}{}^{E} \right) Y_{E0}
	+ \left( r_{A [B} -\pi_{A E} \pi^{E}{}_{[B} +\pi \pi_{A [B}
		-\frac{2\Lambda}{n-2} g_{A[B} \right) Y_{C]0}
	= 0 \;,
\end{dmath}
$\frac{(n-3)}{6}\left.\nabla_0 (\d Y)_{AB0}\right|_\Sigma=0 \colon$
\begin{dmath}\label{eq:CD0dY0BC}
	2(n-2) I_{0[A:B]0}[Y] \hiderel{=}
	\frac{1}{2} \left( 2\pi_A{}^C \pi_B{}^D + r_{AB}{}^{CD}\right) Y_{CD}
	- (n-2) (\pi^{CD} \pi_{C[A} Y_{B]D}
		- \pi \pi_{[A}{}^C Y_{B]C} - r_{[A}{}^C Y_{B]C})
	+ (n-4) D_{[A} \pi_{B]}{}^C Y_{C0}
	+ \frac{2 (n-3)}{(n-2)}  \Lambda Y_{AB}
	= 0 \;.
\end{dmath}
\end{subequations}
\end{theorem}

\begin{proof}
First, we compute the split form of the independent and non-trivial
components of these $\CKY[Y]$ and $\d Y$ operators:
\begin{dgroup}\label{eq:cky-dy-split}
\begin{dmath}\label{eq:cky00C}
	\CKY_{0\colon 0C}[Y]
	= \frac{3}{n-1}\bigg( 
		2 \pi^{F}{}_{[F} Y_{C]0} -  
		D_{F}Y_{C}{}^{F} - 
		(n-2)\nabla_{0}Y_{C0}
		\bigg),
\end{dmath}
\begin{dmath}\label{eq:cky0BC}
	\CKY_{0\colon BC}[Y]
	= 2\CKY_{[B\colon C]0}[Y]
	= 2(\pi_{[B}{}^{A} Y_{C]A}
	- D_{[B}Y_{C]0}
	+ \nabla_{0}Y_{BC}),
\end{dmath}
\begin{dmath}\label{eq:ckyAB0}
	\CKY_{(A\colon B)0}[Y]
	= \frac{3}{2} \left( \bCK_{AB}[Y_{\cdot 0}]
		- 2\pi_{(A}{}^{D} Y_{B)D} \right) ,
\end{dmath}
\begin{dmath}\label{eq:ckyABC}
	\CKY_{A\colon BC}[Y]
	= \bCKY_{A:BC}[Y] \\
	+ 6 \pi_{A[B} Y_{C]0}
	- \frac{6}{n-2} g_{A[B} (\pi g_{C]D} - \pi_{C]D}) Y^{D}{}_0
	+ \frac{2}{n-2} g_{A[B} \CKY_{0\colon 0|C]}[Y] ,
\end{dmath}
\begin{dmath}\label{eq:dyAB0}
	(\d Y)_{AB0}
	= 2 ( 2 D_{[A}Y_{B]0} - 2 \pi_{[A}{}^{C} Y_{B]C} + \nabla_{0}Y_{AB} ) ,
\end{dmath}
\begin{dmath}\label{eq:dyABC}
	(\d Y)_{ABC}
	= (\bd Y)_{ABC} ,
\end{dmath}
\end{dgroup}
where we have used $\bCK$, $\bCKY$, $\bd$ to denote the
$(n-1)$-dimensional versions of the operators defined in~\eqref{eq:ck},
\eqref{eq:cky-def} and~\eqref{eq:dR-def} respectively. We will use
$\CKY_{0:0C}[Y]=0$ and $\CKY_{0:BC}[Y]=0$ to systematically eliminate
$\nabla_0$ derivatives of $Y_{A0}$ and $Y_{AB}$ from the rest of the
calculations. Specifically, this results in the substitutions
\begin{dgroup}\label{eq:cdtime}
\begin{dmath}\label{eq:cdtime-electric}
	\nabla_{0}Y_{A0}
	= \frac{1}{n-2}\left(
			D^{C}Y_{CA} -\pi_{A}{}^{C} Y_{C0} + \pi Y_{A0}
		\right) ,
\end{dmath}
\begin{dmath}\label{eq:cdtime-magnetic}
	\nabla_{0}Y_{AB}
	= D_{[A}Y_{B]0} - \pi_{[B}{}^{C} Y_{A]C} .
\end{dmath}
\end{dgroup}
Similarly, we will use the Einstein equations~\eqref{eq:ricci-split} to
systematically eliminate $\nabla_0 \pi_{AB}$, $\nabla^B \pi_{AB}$ and
the spatial Ricci scalar $r$ throughout our calculations.

Further, the conditions $\nabla_0 \CKY_{0:0C}[Y] = \nabla_0
\CKY_{0:BC}[Y] = 0$ will appear as part of setting to zero the
$\nabla_0$ derivatives of all the components in~\eqref{eq:cky-dy-split}.
However, they can always be satisfied by solving for $\nabla_0^2 Y_{A0}$
and $\nabla_0^2 Y_{AB}$, in analogy with~\eqref{eq:cdtime}. Strictly
speaking, these second order derivatives are constrained by the
propagation equation $Q_{ab}[Y] = 0$. But one of the requirements on the
propagation identities, imposed by Proposition~\ref{prp:propeq} and
verified by Theorem~\ref{thm:ccyk-propeq}, is that $Q[Y]$ factors
through the $\CKY[Y]$ and $\d Y$, which means that we can solve for
$\nabla_0^2 Y_{A0}$ and $\nabla_0^2 Y_{AB}$ just by
differentiating~\eqref{eq:cdtime} and the conditions $\nabla_0
\CKY_{0:0C}[Y] = \nabla_0 \CKY_{0:BC}[Y] = 0$ do not impose any
independent purely spatial constraints on the initial data for $Y_{ab}$.

Now, in the same way that we obtained the combined
form~\eqref{eq:ccyk-combined} of the spacetime cCYK operator, combining
$\CKY_{(A:B)0}[Y]=0$ and $(\d Y)_{AB0}=0$ immediately
gives~\eqref{eq:tfcdelectric}, while combining $\CKY_{A:BC}[Y]=0$ and
$(\d Y)_{ABC}=0$ immediately gives~\eqref{eq:tfcdmagnetic}, the first
two cCYKID conditions.

It now remains to take the $\nabla_0$ derivatives of the already
obtained~\eqref{eq:tfcdelectric} and~\eqref{eq:tfcdmagnetic},
systematically eliminate $\nabla_0 Y_{A0}$ and $\nabla_0 Y_{AB}$ as
above, and to simplify the results (meaning trying to eliminate as many
high order spatial derivatives of $Y$ as possible) using purely spatial
integrability conditions of the same equations. We can shortcut this
process by taking advantage of the spacetime integrability
condition~\eqref{eq:ccyk-ic}. Splitting that identity results in the
following relevant components:
\begin{dgroup}
\begin{dmath}
	\frac{(n-3)}{6(n-2)} \nabla_0 \CKY_{(A\colon B)0}[Y]
		+ O(\CKY[Y]) + O(\d Y)
	= I_{0(A:B)0}[Y] ,
\end{dmath}
\begin{dmath}
	\frac{1}{6} \nabla_0 \CKY_{A\colon BC}[Y]
		+ O(\CKY[Y]) + O(\d Y)
	= I_{0A:BC}[Y] ,
\end{dmath}
\begin{dmath}
	\frac{(n-3)}{12(n-2)} \nabla_0 (\d Y)_{AB0}[Y]
		+ O(\CKY[Y]) + O(\d Y)
	= I_{0[A:B]0}[Y] ,
\end{dmath}
\begin{dmath}
	\frac{1}{12} \nabla_0 (\d Y)_{ABC}[Y]
		+ O(\CKY[Y]) + O(\d Y)
	= I_{0[A:BC]}[Y] ,
\end{dmath}
\end{dgroup}
where $O(-)$ denotes linear dependence on the argument and any of its
spatial derivatives. Each right-hand side is already a zeroth order
operator acting on $Y$. Computing the components of $I_{0(A:B)0}[Y]$ and
$I_{0[A:B]0}[Y]$ directly gives us the desired cCYKID
conditions~\eqref{eq:CD0CYKIDAB0} and~\eqref{eq:CD0dY0BC}. Next, we find
\begin{dmath}\label{eq:CD0dYABC}
	I_{0[A\colon BC]}[Y] \hiderel{=}
	-Y_{[A}{}^{D} D_ {B}\pi_{C]D} \;,
\end{dmath}
which is of spatial tensor type $\yd{1,1,1}$. It happens to be
proportional to an integrability condition obtained by applying $D_C$
to~\eqref{eq:tfcdelectric} and projecting onto~\yd{1,1,1}. Another
integrability condition that we can get is the projection of the
derivative of~\eqref{eq:tfcdelectric} onto spatial tensors of
type~\yd{2,1}\ (in one of the two possible ways), which helps us simplify
the explicit expression for $I_{0A:BC}[Y]$. The resulting simplified
expression is our remaining cCYKID condition~\eqref{eq:CD0CYKIDABC}.

Finally, having established that the cCYKID conditions~\eqref{eq:ccykid}
are equivalent to $\CKY[Y]=0$, $(\d Y) = 0$, and $\nabla_0 \CKY[Y] = 0$,
$\nabla_0 (\d Y) = 0$ on $\Sigma$, a joint application of
Theorem~\ref{thm:ccyk-propeq} and Proposition~\ref{prp:propeq} completes
the proof.
\end{proof}

\begin{remark}\label{rmk:finite-type}\em
It is well-known that the Killing equation and its generalizations,
including Killing-Yano, Killing-St\"ackel equations and their conformal
versions, tend to have at most a finite number of linearly independent
solutions, with the maximal number achieved on maximally symmetric
backgrounds~\cite{Houri_2015}. This behavior is characteristic of PDEs
of so-called \emph{finite type}~\cite[Apx.A]{kh-compat}, which are
defined by the property that the Taylor expansion of a general solution
at any point admits only finitely many independent coefficients. The
finite type property depends only on the symbol of the equaiton (the
coefficients of the highest derivative terms) and it may be enough to
check a subsystem. In fact, considering more equations only increases
the constraints on the number of linearly independent solutions, while
adding subleading terms may only add integrability conditions, which do
the same. So, given that the cCYK equation is of finite type and that
our Theorem~\ref{thm:ccykid} establishes a bijection between solutions
to the cCYKID conditions~\eqref{eq:ccykid} on in initial data surface
$\Sigma$ and solutions to the cCYK equations on the domain of dependence
of $\Sigma$, the cCYKID equations should themselves be of finite type.
Indeed, this can be checked explicitly by noting that the symbol
of~\eqref{eq:tfcdmagnetic} coincides with the symbol of the
form~\eqref{eq:ccyk-combined} of the cCYK equation for $Y_{BC}$ on
$\Sigma$, while the symbol of the $(AB)$ symmetrization
of~\eqref{eq:tfcdelectric} coincides with the symbol of the CK
equation~\eqref{eq:ck} for $Y_{B0}$ on $\Sigma$, which is also
well-known to be of finite type.
\end{remark}

\subsection{Propagation identity in dimension $n=4$} \label{sec:ccyk-4dim}

The $5$-parameter propagation identity from
Theorem~\ref{thm:ccyk-propeq} can be specialized to dimension $n=4$, but
it fails one of the seven inequalities needed to establish hyperbolicity
of the $P$ and $Q$ operators, for any value of the parameters. More
specifically, it fails the inequality associated with the coefficient of
the operator $P^2$. Fortunately, we can use the same trick that was
used for the conformal Killing operator in~\cite{gpkh-ckid}. The idea is
to reduce $P^2$ from a second order to a first order operator by
decoupling a differential consequence of the cCYK system and propagating it
independently. Ultimately, instead of a second order one, we will find a
fourth order propagation identity for the cCYK system in dimension
$n=4$.

Note the following identity (valid in general dimension and without
restriction on the Ricci tensor $R_{ab}$):
\begin{dgroup*}
\begin{dmath} \label{eq:k-to-cky-ident}
	\K_{ab}[\delta Y]
	= \frac{(n-1)}{3(n-2)} S_{ab}[\CKY[Y]]
		+ \frac{(n-1)}{(n-2)} 2 R_{(a}{}^c Y_{b)c} ,
\end{dmath}
\begin{dmath}\label{eq:k-to-cky}
	\text{with} \quad
	S_{ab}[C] \coloneq 2\nabla^c C_{(a:b)c} .
\end{dmath}
\end{dgroup*}
The Ricci-dependent term vanishes for $\Lambda$-vacua, when $R_{ab} =
\frac{2\Lambda}{n-2} g_{ab}$. This identity can be used to factor
\begin{dgroup*}
\begin{dmath} \label{eq:p2-through-k}
	P^2_{a\colon bc}[\CKY[Y]] = -\frac{3(n-2)}{(n-1)}
		\bar{P}_{a:bc}[\K[\delta Y]] ,
\end{dmath}
\begin{dmath} \label{eq:barP-def}
	\text{with} \quad
	\bar{P}_{a\colon bc}[h]
	\coloneq 2\nabla_{[b} h_{c]a}
		- \frac{2}{n-1} g_{a[b} \nabla^d h_{c]d}
		+ \frac{2}{n-1} g_{a[b} \nabla_{c]} h_d{}^d .
\end{dmath}
\end{dgroup*}
But the Killing operator $\K_{ab}[v]$ satisfies its own propagation
identity~\eqref{eq:k-propeq}, which is compatible with that from
Theorem~\ref{thm:ccyk-propeq} in the sense that
\begin{equation} \label{eq:v-propeq-Y}
	{\textstyle (\square + \frac{2\Lambda}{n-2})} (\delta Y)_a
		+ \frac{(n-1)}{3(n-2)} R_{a}{}^{bcd} \CKY_{b:cd}[Y]
	= \frac{1}{y_1} \frac{(n-1)}{3(n-2)} (\delta Q[Y])_a ,
\end{equation}
with $Q[Y]$ defined by Theorem~\ref{thm:ccyk-propeq}. Writing the
propagation identity~\eqref{eq:k-propeq} in terms of $Y$, we get
\begin{dmath} \label{eq:k-propeq-Y}
	\square \K_{ab}[\delta Y]
		- \frac{(n-1)}{3(n-2)} 2R^c{}_{ab}{}^{d} S_{cd}[\CKY[Y]]
		+ \frac{(n-1)}{3(n-2)} \K_{ab}[R\cdot \CKY[Y]]
	= \frac{1}{y_1}\frac{(n-1)}{3(n-2)} \K_{ab}[\delta Q[Y]] ,
\end{dmath}
where $(R\cdot\CKY[Y])_a = R_a{}^{bcd} \CKY_{b:cd}$.

We can now use the same strategy as was used for the \emph{conformal
Killing} equation in~\cite[Sec.4]{gpkh-ckid} to prove
\begin{theorem} \label{thm:ccyk-propeq-4dim}
Under the same hypotheses as Theorem~\ref{thm:ccyk-propeq}, but for
$n=4$, there exists a 6-parameter family of 4th order propagation
identities of the form
\begin{dmath} \label{eq:ccyk-propeq-4dim}
	\square \begin{bmatrix}
		p_1 P^1 + (p_2-y_2) P^2 + p_3 P^3 & p_7 P^7 \\
		\hat{p}_5 \hat{P}^5 & \hat{p}_1 \hat{P}^1 + \hat{p}_2 \hat{P}^2
	\end{bmatrix}
	\begin{bmatrix}
		\CKY \\
		\d
	\end{bmatrix}
		+ \text{l.o.t}
	= \begin{bmatrix}
		\square \sigma^1 - \frac{y_2}{y_1} \bar{P}\circ \K\circ\delta \\
		\square \hat{\sigma}^1
	\end{bmatrix}
	\begin{bmatrix}
		r_1 \rho^1 & r_2 \rho^2
	\end{bmatrix}
	\begin{bmatrix}
		\CKY \\
		\d
	\end{bmatrix} ,
\end{dmath}
where {l.o.t} stands for operators of differential order three or lower
acting on the cCYK system, while the $p_i$, $\hat{p}_j$ and $r_k$
coefficients depend on the free parameters $x$, $y$, $z$, $x_1$, $y_1$
in the same way as in Theorem~\ref{thm:ccyk-propeq} and $y_2$ is an
additional free parameter. The necessary and sufficient conditions for
the generalized normal hyperbolicity of the corresponding $Q$ operator
are still
\begin{subequations} \label{eq:ccyk-propeq-hyp-4dim}
\begin{equation}
	x_1 \ne 0 , \quad
	y_1 \ne 0 ,
\end{equation}
and for the corresponding $P$ operator they are now
\begin{equation}
\begin{gathered}
	x \ne 0, \quad
	y_2 \ne 0 , \quad
	y \ne 0 , \\
	x_1 (3z-x - 6y_1) \ne 0 ,
	\quad \text{and} \quad
	y_1 \ne 0 .
\end{gathered}
\end{equation}
\end{subequations}
\end{theorem}

\begin{proof}
The first step is to apply the wave operator $\square$ to both sides of
the propagation identity from Theorem~\ref{thm:ccyk-propeq} restricted
to $n=4$ dimensions. Then, note that we are completely free to do the
following rewriting:
\begin{dmath}
	\square p_2 P^2 \circ \CKY
	= \square (p_2 - y_2) P^2\circ \CKY
		- y_2 \frac{3(n-2)}{(n-1)} \square \bar{P}\circ \K\circ \delta
	= \square (p_2 - y_2) P^2\circ \CKY
		- y_2 \frac{3(n-2)}{(n-1)} \bar{P}\circ \square\K\circ \delta
		+ \text{l.o.t} .
\end{dmath}
Finally, using~\eqref{eq:k-propeq-Y} to eliminate $\square \K\circ
\delta$ from the above formula, we arrive directly at the desired
propagation identity~\eqref{eq:ccyk-propeq-4dim}. Recalling the relevant
hyperbolicity conditions from Lemmas~\ref{lem:hyp-Q-ccyk}
and~\ref{lem:hyp-P-ccyk}, which are unchanged when the operators
contributing to the principal symbol are multiplied by a power of
$\square$, we get the corresponding
inequalities~\eqref{eq:ccyk-propeq-hyp-4dim}.
\end{proof}

\subsection{Construction of cCYKID in dimension $n=4$} \label{sec:ccykid-4dim}

In contrast to the case of $n>4$ dimensions (Section~\ref{sec:ccykid}),
the fact that in $n=4$ dimensions we must use the \emph{fourth order}
propagation identity from Theorem~\ref{thm:ccyk-propeq-4dim} to apply
Proposition~\ref{prp:propeq} means that the corresponding cCYKID
conditions must be obtained by evaluating $\left.\nabla_0^k
\CKY[Y]\right|_\Sigma = 0$ and $\left.\nabla_0^k \d Y\right|_\Sigma = 0$
for $k=0,1,2,3$. But, our task is simplified by the observation, already
exploited in the proof of Theorem~\ref{thm:ccyk-propeq-4dim}, that the
fourth order identity~\eqref{eq:ccyk-propeq-4dim} follows from the
coupled set of \emph{second order} propagation
identities~\eqref{eq:k-propeq-Y} for $\K[\delta Y]=0$
and~\eqref{eq:ccyk-propeq} for $\CKY[Y]=0$, $\d Y=0$. We have previously
encountered an analogous situation in the construction of the
\emph{conformal Killing initial data}~\cite{gpkh-ckid}. The same
argument as in the proof of Theorem~3 of~\cite{gpkh-ckid}, which we do
not reproduce here, shows that it is in fact sufficient to evaluate the
initial data conditions $\left.\nabla_0^k \CKY[Y]\right|_\Sigma = 0$,
$\left.\nabla_0^k \d Y\right|_\Sigma = 0$ and $\left.\nabla_0^k
\K[\delta Y]\right|_\Sigma = 0$ only for $k=0,1$.

\begin{theorem} \label{thm:ccykid-4dim}
Consider a globally hyperbolic Einstein $\Lambda$-vacuum Lorentzian ma\-ni\-fold,
$(M,g)$ of dimension $n=4$ with $R_{ab}=\Lambda g_{ab}$, and
a Cauchy surface $\Sigma\subset M$. The necessary and sufficient
conditions yielding a set of \emph{closed conformal Killing-Yano initial
data} (cCYKID) for $Y_{ab}$ on $\Sigma$ are the initial data conditions
of Theorem~\ref{thm:ccykid} (specialized to $n=4$) together with the KID
conditions~\eqref{eq:kid} applied to $v=\delta Y$, which can be
rewritten in two equivalent ways, modulo the conditions already included
in Theorem~\ref{thm:ccykid}. The first consists of only
the~\eqref{eq:kid1} condition
\begin{multline} \label{eq:ccykid-4dim-kid1}
	D_{A} D_{B} v_{0}
	+ (2(\pi\cdot\pi)_{A B} - \pi \pi_{AB} - r_{AB}) v_0 \\
	- 2 \pi_{(B}{}^{C} D_{A)} v_C
	- (D^{C}{\pi_{A B}}) v_{C} + \frac{4\Lambda}{n-2}g_{AB} v_0 = 0 ,
\end{multline} 
with
\begin{equation} \label{eq:divY-split}
	v_0 = D^B Y_{B0}, \quad
	v_A = \frac{3}{2} (D^B Y_{BA} - \pi_{A}{}^B Y_{B0} + \pi Y_{A0}) .
\end{equation}
The second equivalent condition is
\begin{multline} \label{eq:kid-dY-equiv}
	12 \pi_{(A|B} r^{BD} Y_{D|C)}
	+ 12 \pi^{BD} r_{(A|B} Y_{D|C)}
	- 12 \pi r_{(A|}{}^B Y_{B|C)}
	\\
	+ 24 \pi \pi^{BD} \pi_{(A|B} Y_{D|C)}
	- 12 \left( \Lambda + \pi_{DE} \pi^{DE} + \pi^2 \right)
		\pi_{(A|}{}^{B} Y_{B|C)}
	\\
	- 12 Y_{(A|B} D_{|C)} D^B \pi
	+ 12 Y_{(A|B} D^E D_E \pi_{|C)}{}^B
	\\
	+ 6 \left[
		5 \pi^{BD} D_D \pi_{AC}
		+ \pi D^B \pi_{AC}
		+ 2 \pi_{AC} D^B \pi
		\right. \\
		- 3 \pi^{BD} D_{(A} \pi_{C)D}
		+ 4 \pi_{D(A} D_{C)} \pi^{DB}
		- \pi D_{(A} \pi_{C)}{}^{B}
		- 2 \pi^B{}_{(A} D_{C)} \pi
		\\ \left.
		- 4 \pi_{(A|}{}^D D_D \pi_{|C)}{}^{B}
		- 2 \pi_{(A}{}^D D^B \pi_{C)D}
	\right] Y_{B0}
	= 0 .
\end{multline}
\end{theorem}
\begin{proof}
As summarized before the statement of the theorem, in imitation of the
proof of~\cite[Thm.3]{gpkh-ckid}, an application of
Proposition~\ref{prp:propeq} implies that the conditions necessary and
sufficient to identify the initial data of a cCYK 2-form $Y_{ab}$ are
equivalent to
\begin{dgroup}
\begin{dmath} \label{eq:cyk-sigma}
	\left.\nabla_0^k \CKY[Y]\right|_\Sigma = 0 ,
\end{dmath}
\begin{dmath} \label{eq:dR-sigma}
	\left.\nabla_0^k \d Y\right|_\Sigma = 0 ,
\end{dmath}
\begin{dmath} \label{eq:K-dY-sigma}
	\left.\nabla_0^k \K[\delta Y]\right|_\Sigma = 0 ,
\end{dmath}
\end{dgroup}
for $k=0,1$. In Theorem~\ref{thm:ccykid}, we have already given a set of
initial data conditions that are intrinsic to $\Sigma$ and are
equivalent to~\eqref{eq:cyk-sigma} and~\eqref{eq:dR-sigma}. Though these
results were stated for $n>4$, all the same calculations remain valid in
dimension $n=4$.

On the other hand, when $n=4$, the propagation
identity~\eqref{eq:ccyk-propeq} fails to be generalized normally
hyperbolic and so~\eqref{eq:cyk-sigma} and~\eqref{eq:dR-sigma} cannot be
used to solve for the $\nabla_0^2\CYK[Y]$ and $\nabla_0^2\d Y$. So these
conditions may no longer be sufficient. Sufficiency is restored by
adding the conditions~\eqref{eq:K-dY-sigma}, which are of course
equivalent to the well-known KID conditions~\eqref{eq:kid} applied to
$v_a = (\delta Y)_a$, whose components specialize
to~\eqref{eq:divY-split} after eliminating $\nabla_0 Y_{A0}$ and
$\nabla_0 Y_{AB}$ using $\CYK_{0:0C} = 0$ and $\CYK_{0:BC} = 0$.

However, these additional KID conditions are not all independent. Namely,
splitting the identity~\eqref{eq:k-to-cky-ident} gives us the schematic
identities
\begin{equation}
\begin{gathered}
	\K_{00}[\delta Y] = O(\CKY[Y]), \quad
	\K_{A0}[\delta Y] = O(\CKY[Y]), \\
	\text{and} \quad
	\K_{AB}[\delta Y] = -\nabla_0 \CKY_{(A:B)0}[Y] + O(\CKY[Y]) ,
\end{gathered}
\end{equation}
where $O(-)$ denotes linear dependence on the argument and any of its
spatial derivatives. Hence, the only independent initial data conditions
will come from $\left.\nabla_0 \K_{AB}[\delta Y]\right|_\Sigma = 0$ or
only the~\eqref{eq:kid1} part of the KID conditions, which we have
copied to~\eqref{eq:ccykid-4dim-kid1} in the statement of the theorem.
Equivalently, as can be seen from the preceding identities, this
remaining independent condition can be replaced by $\left.\nabla_0^2
\CKY_{(A:B)0}\right|_\Sigma = 0$. In the proof of
Theorem~\ref{thm:ccykid}, we have already shown that the condition
$\left.\nabla_0 \CKY_{(A:B)0}\right|_\Sigma = 0$ is equivalent
to~\eqref{eq:CD0CYKIDAB0}, which no longer contains any spatial
derivatives of $Y_{ab}$. Thus, to obtain the new independent condition
on $Y_{ab}$, it is sufficient to apply $\nabla_0$
to~\eqref{eq:CD0CYKIDAB0} and once again eliminate all $\nabla_0
Y_{ab}$. In this way, while also eliminating $\nabla_0 \pi_{AB}$ using
the Einstein equations~\eqref{eq:ricci-split} and $\nabla_0 r_{ABCD}$
using~\eqref{eq:driem}, direct calculation gives us the desired initial
data condition~\eqref{eq:kid-dY-equiv}.
\end{proof}

\section{Conformal Killing Yano initial data}
\label{sec:cykid}
In a general dimension $n$ it is not yet known how to construct a
propagation identity for the conformal Killing-Yano (CYK) system,
without the closed condition that was used in the successful
construction in Section~\ref{sec:ccyk} on an Einstein
($\Lambda$-vacuum) background. But, as we will analyze in this section,
the problem can be solved if $n=4$. In principle, the solution can be
extracted from the previously studied case of the Killing
$(2,0)$-spinor~\cite{GOMEZLOBO2008, bk-kerrness1, bk-kerrness2}, which is the spinorial version of a
self-dual conformal Killing-Yano 2-form. Instead, we give a purely
tensorial derivation, taking advantage of the explicit calculations from
Section~\ref{sec:ccyk} and the conceptually clear approach to the
problem that we have described in Section~\ref{sec:propeq} and our
previous work~\cite{gpkh-ckid}. Below, we will freely use the notation
and results introduced in Sections~\ref{sec:propeq}
and~\ref{sec:ccyk}.

Recall formula~\eqref{eq:k-to-cky-ident}, which factors $\K_{ab}[\delta
Y]$ through $\CKY_{a:bc}[Y]$ in general dimension. And also note the
following formula, which is valid only in 4 dimensions (but without
restriction on the Ricci tensor $R_{ab}$):
\begin{dgroup*}
\begin{dmath}\label{eq:Kdual-to-CKY}
	\K_{ab}[{*}\d Y]
	= 3 \bar{S}_{ab}[\CKY[Y]]
		- 18 R_{(a}^c \eta_{b)c}{}^{de} Y_{de} .
\end{dmath}
\begin{dmath}
	\text{with} \quad
	\bar{S}_{ab}[C]
	= 2 \eta_{(a|}{}^{cde} \nabla_e C_{|b):cd} .
\end{dmath}
\end{dgroup*}
The Ricci-dependent term vanishes for $\Lambda$-vacua, when $R_{ab} =
\frac{2\Lambda}{n-2} g_{ab}$. For reference, our $4$-dimensional
conventions for the Hodge $*$ operation are
\begin{equation}
	({*}v)_{abc} = \eta_{abc}{}^d v_d , \quad
	({*}Y)_{ab} = \frac{1}{2} \eta_{ab}{}^{cd} Y_{cd} , \quad
	({*}w)_{a} = \frac{1}{6} \eta_{a}{}^{bcd} w_{bcd} ,
\end{equation}
with $\eta_{abcd}$ being the Levi-Civita tensor.

Recall also that the general 5-parameter propagation identity can be
specialized to both $n=4$ dimensions and also to the case which
decouples the $\CKY[Y]$ operator from the exterior derivative $\d Y$
(setting $p_7 = p_8 = r_2 = 0$). But in both cases at least one of the
inequalities from the hyperbolicity
conditions~\eqref{eq:ccyk-propeq-hyp-4dim} fails. In
Section~\ref{sec:ccyk-4dim}, this failure when $n=4$ was fixed by
decoupling $\K_{ab}[\delta Y]$ operator and propagating it separately.
By analogy, in this section, we will fix the failure of hyperbolicity
for the $n=4$ propagation identity decoupled from $\d Y$, by also
decoupling $\K_{ab}[{*}\d Y]$ and propagating it separately. In general
dimension, restoring the hyperbolicity for the decoupled case remains an
open problem.

By decoupling the general 5-parameter propagation cCYK identity
(Theorem~\ref{thm:ccyk-propeq}) from $\d Y$ (setting $p_7 = p_8 = r_2 =
0$), what remains is the following 1-parameter identity
\begin{equation} \label{eq:cyk-propeq2-old}
	{-P_3} = \sigma^1 \circ \rho^1 \circ \CKY ,
\end{equation}
where the single parameter is just an overall multiplicative constant.
To apply the decoupling strategy in Section~\ref{sec:ccyk-4dim}, we used
the fact that $P^2_{a:bc}[\CKY[Y]]$ factors through $\K_{ab}[\delta Y]$,
according to~\eqref{eq:p2-through-k}. Unfortunately,
$P^3_{a:bc}[\CKY[Y]]$ does not directly factor through $\K_{ab}[{*}\d
Y]$. Thus, it is convenient to introduce the alternative operator
$\bar{P}^3$, which does factor:
\begin{equation}
	\bar{P}^3_{a:bc}[\CKY[Y]] := \frac{3}{2} ({*}\bar{P}[\K[{*}\d Y]])_{a:bc} ,
\end{equation}
where the operator $\bar{P}$ was defined in~\eqref{eq:barP-def} and we
have extended the Hodge $*$ operator to
\begin{equation}
	({*}C)_{a:bc}[Y]
	:= \frac{2}{3} (\eta_{bc}{}^{pq} C_{a:pq}
		+ \eta_{a[b}{}^{pq} C_{c]:pq}) .
\end{equation}
With this choice, identity~\eqref{eq:cyk-propeq2-old} gets rewritten as
\begin{equation} \label{eq:cyk-propeq2-4dim}
	3 P^1 - \frac{3}{2} P^2 - \frac{1}{2} \bar{P}^3
		+ \frac{1}{2} P^4 - 3 P^5 - 5 P^6
	= \sigma^1 \circ \rho^1 \circ \CKY .
\end{equation}
To clarify the structure of this identity, let us rewrite it more
explicitly as
\begin{multline}
	3 \square \CKY_{a:bc}[Y]
	+ 3 \bar{P}_{a:bc}[\K[\delta Y]]
	- \frac{3}{4} ({*}\bar{P}[\K[{*}\d Y]])_{a:bc}
	+ l.o.t \\
	= \CKY_{a:bc}[Q[Y]] ,
\end{multline}
where both $P^1_{a:bc}[C] = \square C_{a:bc}$ and $Q[Y] = \rho^1\circ
\CKY[Y]$ are generalized normally hyperbolic in the required way
(Theorem~\ref{thm:ccyk-propeq-4dim}), while the $P^2$ and $\bar{P}^3$
terms have been rewritten as first order operators on $\K[\delta]$ and
$\K[{*}\d Y]$.

Before proceeding, let us specialize identity~\eqref{eq:v-propeq-Y} to
our choice of dimension and parameters ($n=4$, $y_1=1$), which shows
that the $\K_{ab}[\delta Y]$ propagates in a way compatible with our
choice of $Q[Y]$:
\begin{dmath} \label{eq:v-propeq-Y-4dim}
	(\square + \Lambda) (\delta Y)_a
		+ \frac{1}{2} R_{a}{}^{bcd} \CKY_{b\colon cd}[Y]
	= \frac{1}{2} (\delta Q[Y])_a ,
\end{dmath}
where $(R\cdot\CKY[Y])_a = R_a{}^{bcd} \CKY_{b:cd}$. Similarly, we must
verify that $\K_{ab}[{*}\d Y]$ also propagates in a way that is
compatible with $Q[Y]$. A direct calculation shows that
\begin{dmath} \label{eq:vdual-propeq-Y-4dim}
	(\square + \Lambda) ({*} \d Y)_a
		+  R_{a}{}^{bcd} {*}\CKY_{b\colon cd}[Y]
	= \frac{1}{2} ({*} \d Q[Y])_a .
\end{dmath}
Alternatively, substituting $Y\mapsto {*}Y$
into~\eqref{eq:v-propeq-Y-4dim}, we immediately
get~\eqref{eq:vdual-propeq-Y-4dim}, after using the following helpful
identities:
\begin{align}
	(\delta {*}Y)_a &= \frac{1}{2} ({*} \d Y)_a , \\
	\CKY_{a:bc}[{*}Y] &= {*}\CKY_{a:bc}[Y] , \\
	Q_{ab}[{*}Y] &= ({*}Q[Y])_{ab} .
\end{align}
Thus, rewriting the Killing equation propagation
identity~\eqref{eq:k-propeq} adapted to our situation, we get
\begin{dgroup*}
\begin{dmath} \label{eq:k-propeq-Y-4dim}
	\square \K_{ab}[\delta Y]
		- R^c{}_{ab}{}^{d} S_{cd}[\CKY[Y]]
		+ \frac{1}{2} \K_{ab}[R\cdot \CKY[Y]]
	= \frac{1}{2} \K_{ab}[\delta Q[Y]] ,
\end{dmath}
\begin{dmath} \label{eq:kdual-propeq-Y-4dim}
	\square \K_{ab}[{*}\d Y]
		- 2 R^c{}_{ab}{}^{d} S_{cd}[{*}\CKY[Y]]
		+ \K_{ab}[R\cdot {*}\CKY[Y]]
	= \frac{1}{2} \K_{ab}[{*}\d Q[Y]] ,
\end{dmath}
\end{dgroup*}
where again $(R\cdot C)_a = R_a{}^{bcd} C_{b:cd}$.

\begin{theorem} \label{thm:cyk-propeq-4dim}
Under the same hypotheses as Theorem~\ref{thm:ccyk-propeq}, but for
$n=4$, there exists the following 2-parameter family of 4th order propagation
identities of the form
\begin{dmath} \label{eq:cyk-propeq-4dim} \textstyle
	\square (3P^1 + (y_2 - 3) \frac{1}{2} P^2 + (y_3-1) \frac{1}{2} \bar{P}^3) \circ \CKY
		+ \text{l.o.t}
	= \left(\square\sigma^1
		+ \frac{1}{2} y_2 \bar{P}\circ \K \circ \delta
		+ \frac{3}{8} y_3 {*}\bar{P}\circ \K \circ {*}\d\right) \circ \rho^1 \circ \CKY ,
\end{dmath}
where {l.o.t} stands for operators of differential order three or lower
acting on the CYK operator and $y_2$, $y_3$ are free parameters. The
necessary and sufficient conditions for the generalized
normal hyperbolicity of the corresponding $P$ operator are
\begin{equation} \label{eq:cyk-propeq-hyp-4dim}
	y_3 \ne 0, \quad y_2 \ne 0 , \quad \text{and} \quad y_3 \ne -2 .
\end{equation}
and the corresponding $Q = \rho^1\circ \CKY$ operator is always
normally-hyperbolic (being independent of the free parameters).
\end{theorem}

The proof is directly analogous to that of
Theorem~\ref{thm:ccyk-propeq-4dim}.

\begin{proof}
The first step is to apply the wave operator $\square$ to both sides of
the propagation identity~\eqref{eq:cyk-propeq2-4dim}. Then, note that we
are completely free to do the following rewriting:
\begin{dgroup*}
\begin{dmath} \textstyle
	\square (-\frac{3}{2}) P^2 \circ \CKY
	= \square (y_2-3)\frac{1}{2} P^2 \circ \CKY
		- y_2 \square \bar{P}\circ \K \circ \delta
	= \square (y_2-3)\frac{1}{2} P^2 \circ \CKY
		- y_2\bar{P}\circ  \square\K \circ \delta
		+ \text{l.o.t} ,
\end{dmath}
\begin{dmath} \textstyle
	\square (-\frac{1}{2}) \bar{P}^3 \circ \CKY
	= \square (y_3-1)\frac{1}{2} \bar{P}^3 \circ \CKY
		- \frac{3}{4} y_3 \square {*}\bar{P}\circ \K \circ {*}\d
	= \square (y_3-1)\frac{1}{2} \bar{P}^3 \circ \CKY
		- \frac{3}{4} y_3 {*}\bar{P}\circ \square\K \circ {*}\d
		+ \text{l.o.t} .
\end{dmath}
\end{dgroup*}
Finally, using~\eqref{eq:k-propeq-Y-4dim}
and~\eqref{eq:kdual-propeq-Y-4dim} to eliminate $\square\K \circ \delta$
and $\square\K \circ {*}\d$ from the above formulas, we arrive directly
at the desired propagation identity~\eqref{eq:cyk-propeq-4dim}.
Recalling the relevant hyperbolicity conditions from
Lemmas~\ref{lem:hyp-Q-ccyk} and~\ref{lem:hyp-P-ccyk} (the latter lemma
is adapted by setting $x=y=0$ and $w=1$ to decouple the $\d$ and $\CKY$
operators, while the translation from the $\bar{P}^3$ to the $P^3$
operator is done by comparing~\eqref{eq:cyk-propeq2-old}
and~\eqref{eq:cyk-propeq2-4dim}), which are unchanged when the operators
contributing to the principal symbol are multiplied by a power of
$\square$, we get the corresponding
inequalities~\eqref{eq:cyk-propeq-hyp-4dim}.
\end{proof}

Next we use the previous propagation equations to construct conformal
Killing-Yano initial data (CYKID) in dimension four.

\begin{theorem}\label{thm:cykid}\label{theo:CYKID}
Consider a globally hyperbolic Einstein $\Lambda$-vacuum Lorentzian ma\-ni\-fold,
$(M,g)$ of dimension $n=4$ with $R_{ab}=\Lambda g_{ab}$, and a Cauchy
hypersurface $\Sigma\subset M$. The necessary and sufficient conditions
yielding a set of \emph{conformal Killing-Yano initial data} (CYKID) for
$Y_{ab}$ on $\Sigma$ are the following equations, where we indicate the
provenance of each of them. Each equality holds modulo the ones
preceding it.

\medskip
\begin{subequations}\label{eq:cykid-4d}
\noindent
$\frac{2}{3}\left.\CKY_{(A:B)0}[Y]\right|_\Sigma=0 \colon$
\quad {\yd{2}} part of \eqref{eq:tfcdelectric} for $n=4$, or
\begin{dmath}\label{eq:cykid-4d-AB0}
	\bCK[Y_{\cdot 0}]_{AB} - 2 \pi_{(A}{}^C Y_{B)C} = 0,
\end{dmath}
$\left.\CKY_{A:BC}[Y]\right|_\Sigma=0 \colon$
\quad {\yd{2,1}} part of \eqref{eq:tfcdmagnetic} for $n=4$, or
\begin{dmath}\label{eq:cykid-4d-ABC}
	\bCKY[Y]_{A\colon BC}
	-6\pi_{A [B} Y_{0 C]}
	-3 g_{A [B} \pi_{C]}{}^{E} Y_{0 E} %\\
	+ 3 \pi g_{A [B} Y_{0 C]} = 0. %\\
\end{dmath}
$\frac{1}{6} \left.\nabla_0 \CKY_{(A:B)0}[Y]\right|_\Sigma=0 \colon$
\quad \eqref{eq:CD0CYKIDAB0} for $n=4$, or
\begin{dmath}\label{eq:CD0CYKIDAB0-4dim}
	%2 I_{0(A\colon B)0}[Y] \hiderel{=} 
	(D_{C}\pi_{AB} - D_{(A}\pi_{B)C}) Y^{C}{}_{0}
	\\
	+ \pi \pi_{(A}{}^{C} Y_{B)C} 
	+ \pi_{(A|C} \pi^{CD} Y_{D|B)}
	+ r_{(A}{}^{C} Y_{B)C} 
	= 0 \;,
\end{dmath}
$-\frac{1}{6} \left.\nabla_0 \CKY_{A:BC}[Y]\right|_\Sigma=0 \colon$
\quad \eqref{eq:CD0CYKIDABC} for $n=4$, or
\begin{dmath}\label{eq:CD0CYKIDABC-4dim}
	%-I_{0A\colon BC}[Y] \hiderel{=}
	(D_{[C}{\pi_{B]}{}^{E}}) Y_{A E}
	+(D_{[C|}{\pi_{A}{}^{E}}) Y_{|B] E}
	\\
	+(D^{E}{\pi_{A [B}}) Y_{C] E}
	+\frac{1}{2}\left( r_{B C A}{}^{E} + 2\pi_{A [B} \pi_{C]}{}^{E} \right) Y_{E0}
	+ \left( r_{A [B} -\pi_{A E} \pi^{E}{}_{[B} +\pi \pi_{A [B}
		-\frac{2\Lambda}{n-2} g_{A[B} \right) Y_{C]0}
	= 0 \;,
\end{dmath}
$\left.\nabla_0\K_{AB}[\delta Y]\right|_\Sigma=0 \colon$
\begin{dmath} \label{eq:nabla0-kid-4d}
	D_{A} D_{B} v_{0}
	+ (2(\pi\cdot\pi)_{A B} - \pi \pi_{AB} - r_{AB}) v_0 \\
	- 2 \pi_{(B}{}^{C} D_{A)} v_C
	- (D^{C}{\pi_{A B}}) v_{C} + 2\Lambda g_{AB} v_0 = 0 ,
\end{dmath}
$\left.\nabla_0\K_{AB}[*\d Y]\right|_\Sigma=0 \colon$
\begin{dmath} \label{eq:nabla0-kiddual-4d}
	D_{A} D_{B} (*w)_{0}
	+ (2(\pi\cdot\pi)_{A B} - \pi \pi_{AB} - r_{AB}) (*w)_0 \\
	- 2 \pi_{(B}{}^{C} D_{A)} (*w)_C
	- (D^{C}{\pi_{A B}}) (*w)_{C} + 2\Lambda g_{AB} (*w)_0 = 0 ,
\end{dmath}
\end{subequations} 
where we have used the same spatial differential operators $\bCK$ and
$\bCKY$ as in the proof of Theorem \ref{thm:ccykid}, and the components
$v_0, v_A$ of $v_a = (\delta Y)_a$ are the same as
in~\eqref{eq:divY-split}, while the $(*w)_0, (*w)_A$ components of
$(*w)_a = ({*}\d Y)_a$ are
\begin{equation}\label{eq:v-dual-split-4d}
	(*w)_{0} \coloneq \eps_{ABC} D^{C}Y^{AB} , \quad
	(*w)_{A} \coloneq  
	\eps_{ABC} \left(\pi^{DB} Y_{D}{}^{C} + 3 D^{B}Y^{C}{}_{0}\right) ,
\end{equation}
with $\eps_{ABC} := \eta_{0ABC}$.
\end{theorem}

The proof is directly analogous to that of Theorem~\ref{thm:ccykid-4dim}.

\begin{proof}
To construct the CYKID conditions, it is sufficient to, once again,
apply Proposition~\ref{prp:propeq} to the 4th order propagation
identity~\eqref{eq:cyk-propeq-4dim} for the $\CYK[Y]$ operator obtained
in Theorem~\ref{thm:cyk-propeq-4dim}. But, more practically, following
the logic explained in the proof of Theorem~\ref{thm:ccykid-4dim}, since
the 4th order identity was obtained by combining compatible second order
propagation identities for the $\CYK[Y]$, $\K[\delta Y]$ and $\K[{*}\d
Y]$ operators, the following initial data conditions are sufficient:
\begin{dgroup}
\begin{dmath} \label{eq:cyk-sigma-4} 
	\left.\nabla_0^k \CKY[Y]\right|_\Sigma = 0 ,
\end{dmath}
\begin{dmath} \label{eq:K-dY-sigma-4} 
	\left.\nabla_0^k \K[\delta Y]\right|_\Sigma = 0 ,
\end{dmath}
\begin{dmath} \label{eq:K-dualdY-sigma-4} 
	\left.\nabla_0^k \K[*\d Y]\right|_\Sigma = 0 ,
\end{dmath}
\end{dgroup}
for $k=0$, $1$.

We can get~\eqref{eq:cykid-4d-AB0} and~\eqref{eq:cykid-4d-ABC} by
computing $\CKY_{(A:B)0}[Y]$ and $\CKY_{A:BC}[Y]$ from the formulas
~\eqref{eq:cky-dy-split}, or just take the appropriately symmetrized
projections of~\eqref{eq:tfcdelectric} and~\eqref{eq:tfcdmagnetic},
specialized to $n=4$, as indicated in the theorem. The remaining
components allow us to systematically eliminate any $\nabla_0$
derivatives of $Y_{A0}$ and $Y_{AB}$ in the remaining calculations. The
computation of $\nabla_0\CKY_{(A:B)0}[Y]$ and $\nabla_0\CKY_{A:BC}[Y]$
follows the same logic as in the proof of Theorem~\ref{thm:ccykid},
producing~\eqref{eq:CD0CYKIDAB0-4dim} and~\eqref{eq:CD0CYKIDABC-4dim} as
indicated in the theorem.

Again following the logic of the proof of Theorem~\ref{thm:ccykid-4dim},
splitting the identities~\eqref{eq:k-to-cky-ident}
and~\eqref{eq:Kdual-to-CKY} and systematically eliminating all
$\nabla_0$ derivatives of $Y_{A0}$ and $Y_{AB}$, allows us to write
\begin{equation}\label{eq:K[deltaY]-4d}
\begin{gathered}
	\K_{00}[\delta Y] = O(\CKY[Y]), \quad
	\K_{A0}[\delta Y] = O(\CKY[Y]), \\
	\K_{AB}[\delta Y] = -\nabla_0 \CKY_{(A:B)0}[Y] + O(\CKY[Y]) ,
\end{gathered}
\end{equation}
\begin{equation}\label{eq:K[*dY]-4d}
\begin{gathered}
	\K_{00}[*\d Y] = O(\CKY[Y]), \quad
	%\K_{A0}[*\d Y] = -3 \varepsilon_A{}^{CD} \nabla_{0}CKY[Y]_{0:CD} + O(\CKY[Y]), \\
	\K_{A0}[*\d Y] = O(\CKY[Y]), \\
	\K_{AB}[*\d Y] = -6 \varepsilon^{CDE} g_{C(A} \nabla_{0}\CKY[Y]_{B):DE}
		+ O(\CKY[Y]) ,
\end{gathered}
\end{equation}
where, as usual, $O(-)$ denotes linear dependence on the argument and
any of its spatial derivatives. Hence, setting any of the above
expressions to zero does not add any new independent initial data
conditions. Obviously, the same will be true of $\nabla_0\K_{00}[\delta
Y]$, $\nabla_0\K_{A0}[\delta Y]$, and $\nabla_0\K_{00}[{*}\d Y]$ and
$\nabla_0\K_{A0}[{*}\d Y]$.

It remains only to compute the initial data conditions from
$\nabla_0\K_{AB}[\delta Y]$ and $\nabla_0\K_{AB}[{*}\d Y]$. Again, as in
the proof of Theorem~\ref{thm:ccykid-4dim}, we know that it would be
sufficient to plug into the second KID condition~\eqref{eq:kid1} the
vectors $v_a = (\delta Y)_a$ and $(*w)_a = ({*}\d Y)_a$, whose split
components, with $\nabla_0$ derivatives eliminated, are by direct
computation given by~\eqref{eq:divY-split}
and~\eqref{eq:v-dual-split-4d} respectively. The result gives us the
remaining CYKID conditions~\eqref{eq:nabla0-kid-4d}
and~\eqref{eq:nabla0-kiddual-4d}. The resulting expressions are third
order spatial differential operators on $Y_{A0}$ and $Y_{AB}$. We can
reduce them to first order differential operators by simply applying
$\nabla_0$ to~\eqref{eq:K[deltaY]-4d} and~\eqref{eq:K[*dY]-4d},
respectively substituting~\eqref{eq:CD0CYKIDAB0-4dim}
and~\eqref{eq:CD0CYKIDABC-4dim} for $\nabla_0 \CKY_{(A:B)0}[Y]$ and
$\nabla_0 \CKY_{A:BC}[Y]$, and systematically eliminating $\nabla_0
Y_{A0}$ and $\nabla_0 Y_{AB}$. However the resulting expressions become
rather long and unenlightening, so we omit them.
\end{proof}

Remark~\ref{rmk:finite-type} applies equally well to check the finite
type property of the 4-dimensional CYKID conditions~\eqref{eq:cykid-4d};
it is sufficient to look at the symbols of~\eqref{eq:cykid-4d-AB0}
and~\eqref{eq:cykid-4d-ABC}.

\section{Discussion}

We derived a set of necessary and sufficient conditions (the cCYKID
equations) ensuring that a $\Lambda$-vacuum initial data set for the
Einstein equations admits a closed conformal Killing-Yano 2-form
(Theorem~\ref{thm:ccykid} in dimensions $n>4$, and
Theorem~\ref{thm:ccykid-4dim} for $n=4$) or, in the special dimension
$n=4$, just a conformal Killing-Yano 2-form (Theorem~\ref{thm:cykid}).
These initial data equations include both differential conditions on the
spatial components $Y_{A0}$ and $Y_{AB}$ of the spacetime 2-form
$Y_{ab}$, as well as purely algebraic conditions that involve the
intrinsic and extrinsic geometry of the initial data surface. While
these results are special to Lorentzian signature, the propagation
identities (Theorems~\ref{thm:ccyk-propeq}, \ref{thm:ccyk-propeq-4dim}
and~\ref{thm:cyk-propeq-4dim}) that we have used to derive the initial
data conditions are fully covariant and hence remain valid in any
pseudo-Riemannian signature. The method that we have used to arrive at
Theorem~\ref{thm:ccyk-propeq} is a representation-theoretic exhaustive
search based on covariance and fixed total degree of various
differential operators. In fact, we have shown that, in dimensions $n>4$
(excluding $n=6$), the result is guaranteed to be the most general one.
As a result, we have also concluded
(Corollary~\ref{cor:cyk-propeq-nogo}) that there does not exist a second
order covariant propagation identity for non-closed CYK 2-forms.

As we have indicated in the Introduction, the ability to describe
$\Lambda$-vacuum Einstein initial data giving rise to a cCYK (or, for
$n=4$, also CYK) 2-form can improve the initial data characterization given
in~\cite{gpgl-kerr} of the Kerr rotating
black hole solution. In higher dimensions, the same idea could be used
to give the first initial data characterizations of
members of the Kerr-NUT-(A)dS rotating black holes. It would be
interesting to explore these possibilities in future work.
As discussed in Remark \ref{rem:isometric-embedding} these prospective
results admit the interpretation of necessary and sufficient conditions
for the existence of isometric embeddings of Riemannian manifolds
in the corresponding ambient spacetimes. Moreover, a general formulation of 
the non-linear stability problem for the Kerr-NUT-(A)dS rotating black holes
must include the idea of general vacuum initial data close in some 
topology to general Kerr-NUT-(A)dS data, constructed from the cCYKID initial 
data characterization of Theorem \ref{theo:CYKID}. 

At the moment, no propagation identity is known for the higher
dimensional CYK 2-forms ($n>4$) or for higher rank cCYK $p$-forms
($p>2$). It would be interesting to study these equations using the
approaches used in this work: representation-theoretic exhaustive search
and clever decoupling of independently propagated integrability
conditions.

\paragraph{Acknowledgements}
IK was partially supported by the Praemium Academiae of M.~Markl,
GA\v{C}R project GA19-06357S and RVO: 67985840.
AGP is supported by grant GA19-01850S of the Czech Science Foundation.

\appendix

\section{Tensor bases} \label{app:tensors}

In this appendix, we list several sequences of tensor-valued covariant
differential operators which, according to the representation-theoretic
discussion in Section~\ref{sec:ccyk} span the space of operators of a
certain total order and tensor type. The notation for these operators
follows tables preceding Remark~\ref{rmk:littlewood}. We presume
throughout that the cosmological vacuum Einstein equations hold, $R_{ab}
= \frac{2\Lambda}{n-2} g_{ab}$, so that the Ricci tensor never appears
in the formulas below.

\begin{remark} \label{rmk:weyl}\em
In the representation-theoretic discussion of Section~\ref{sec:ccyk}, we
noted that, being traceless, the Weyl tensor $W_{abcd}$ is precisely of
representation type $\yd{2,2}$. The Riemann tensor $R_{abcd}$ has all
the same symmetries, but is not traceless, due to a possibly
non-vanishing cosmological constant $\Lambda$. Thus, following
strict representation-theoretic logic, we should write all terms
involving $W_{abcd}$ and $\Lambda$ separately, like so:
\begin{equation}
	c_1 O(W_{abcd}) + c_2' O(\Lambda) + \cdots .
\end{equation}
However, expressing $W_{abcd}$ in terms $R_{abcd}$, $g_{ab}$ and
$\Lambda$, any such expression becomes
\begin{equation}
	c_1 O(R_{abcd}) + c_2 O(\Lambda) + \cdots ,
\end{equation}
where $c_1$ stays the same but $c_2$ may now be different. Since for the
purposes of computer algebra it is more economical to work directly with
$R_{abcd}$, rather than $W_{abcd}$, we choose to work with the
coefficients $c_1$ and $c_2$ in the second formulation directly.
\end{remark}

The following is a basis of the possible second total order covariant
differential operators of type $\yd{1,1}\,Y \to \yd{1,1}\,Q$:
\begin{align} \label{eq:Q-def}
	Q^1_{ab}[Y] &= \square Y_{ab} , \\
	Q^2_{ab}[Y] &= 2\nabla_{[a} \nabla^d Y_{b]d} = -(\d\delta Y)_{ab} , \\
	Q^3_{ab}[Y] &= R_{ab}{}^{de} Y_{de} , \\
	Q^4_{ab}[Y] &= \Lambda Y_{ab} .
\end{align}

The following are bases of the possible second total order covariant
differential operators of type $\yd{2,1}\,C \to \yd{2,1}\,P$ and
$\yd{1,1,1}\,\Xi \to \yd{2,1}\,P$:
\begin{align} \label{eq:P-def}
	P^1_{a:bc}[C] &= \square C_{a:bc} , \\
\notag
	P^2_{a:bc}[C] &= \nabla_b \nabla^d (C_{a:dc}+C_{c:da})
		- \nabla_c \nabla^d (C_{a:db}+C_{b:da}) \\
	&\quad {}
		-\frac{2}{n-1} g_{a[b|} \nabla^e\nabla^f C_{(e:f)|c]}
		+ \frac{3}{2(n-1)} g_{a[b} R_{c]}{}^{def} C_{d:ef} , \\
\notag
	P^3_{a:bc}[C] &= 2\nabla_a \nabla^d C_{d:cb}
		- \nabla_b \nabla^d C_{d:ac} + \nabla_c \nabla^d C_{d:ab} \\
	&\quad {}
		+ \frac{6}{n-1} g_{a[b|} \nabla^e\nabla^f C_{(e:f)|c]}
		+ \frac{3}{2(n-1)} g_{a[b} R_{c]}{}^{def} C_{d:ef} , \\
\notag
	P^4_{a:bc}[C] &= 2 R_{bc}{}^{ef} C_{a:ef}
		- R_{ca}{}^{ef} C_{b:ef} + R_{ba}{}^{ef} C_{c:ef} \\
	&\quad {}
		+ \frac{6}{(n-1)} g_{a[b} R_{c]}{}^{def} C_{d:ef} , \\
\notag
	P^5_{a:bc}[C] &= R_a{}^e{}_b{}^f (C_{e:cf}+C_{f:ce})
		- R_a{}^e{}_c{}^f (C_{e:bf}+C_{f:be}) \\
	&\quad {}
		+ \frac{3}{(n-1)} g_{a[b} R_{c]}{}^{def} C_{d:ef} , \\
	P^6_{a:bc}[C] &= \Lambda C_{a:bc} ; \\
	P^7_{a:bc}[\Xi]
	&= \CKY_{a:bc}[\delta\Xi] , \\
	P^8_{a:bc}[\Xi]
	&= \frac{1}{6} (2R_{bc}{}^{de} \Xi_{a de}
		- R_{ab}{}^{de} \Xi_{c de}
		+ R_{ac}{}^{de} \Xi_{b de}) .
\end{align}

The following are bases of the possible second total order covariant
differential operators of type $\yd{1,1,1}\,\Xi \to \yd{1,1,1}\,\hat{P}$
and $\yd{2,1}\,C \to \yd{1,1,1}\,\hat{P}$:
\begin{align} \label{eq:hatP-def}
	\hat{P}^1_{abc}[\Xi] &= \square \Xi_{abc} , \\
	\hat{P}^2_{abc}[\Xi] &= -(d\delta\Xi)_{abc} , \\
	\hat{P}^3_{abc}[\Xi]
		&= R_{ab}{}^{de} \Xi_{c de}
		 + R_{bc}{}^{de} \Xi_{a de}
		 + R_{ca}{}^{de} \Xi_{b de} , \\
	\hat{P}^4_{abc}[\Xi] &= \Lambda \Xi_{abc} ; \\
	\hat{P}^5_{abc}[C]
		&= \nabla_a\nabla^d C_{d:bc}
		 + \nabla_b\nabla^d C_{d:ca}
		 + \nabla_c\nabla^d C_{d:ab} , \\
	\hat{P}^6_{abc}[C]
		&= R_{ab}{}^{de} C_{c:de}
		 + R_{bc}{}^{de} C_{a:de}
		 + R_{ca}{}^{de} C_{b:de} .
\end{align}

The following is a basis of the possible third total order covariant
differential operators of type $\yd{1,1}\,Y \to \yd{2,1}\,T$:
\begin{align} \label{eq:T-def}
\notag
	T^1_{a:bc}[Y] &= \square \CKY_{a:bc}[Y] , \\
\notag
	T^2_{a:bc}[Y]
	&= \nabla_{(a}\nabla_{b)} (\delta Y)_c
	  -\nabla_{(a}\nabla_{c)} (\delta Y)_b
	  -\frac{2}{n-1} g_{a[b} \square (\delta Y)_{c]} \\
	&\quad {}
	  +\frac{2\Lambda}{(n-1)(n-2)} g_{a[b} (\delta Y)_{c]} , \\
\notag
	T^3_{a:bc}[Y] &= 2 R_{bc}{}^{ef} \CKY_{a:ef}[Y]
		- R_{ca}{}^{ef} \CKY_{b:ef}[Y] + R_{ba}{}^{ef} \CKY_{c:ef}[Y] \\
	&\quad {}
		+ \frac{6}{n-1} g_{a[b} R_{c]}{}^{def} \CKY_{d:ef}[Y] , \\
\notag
	T^4_{a:bc}[Y] &= 2 R_a{}^{(e}{}_b{}^{f)} \CKY_{e:cf}[Y]
		- 2 R_a{}^{(e}{}_c{}^{f)} \CKY_{e:bf}[Y] \\
	&\quad {}
		+ \frac{3}{(n-1)} g_{a[b} R_{c]}{}^{def} \CKY_{d:ef}[Y] , \\
	T^5_{a:bc}[Y]
	&=  \frac{2}{6} R_{bc}{}^{de} (\d Y)_{ade}
		- \frac{1}{6} R_{ab}{}^{de} (\d Y)_{cde}
		+ \frac{1}{6} R_{ac}{}^{de} (\d Y)_{bde} , \\
	T^6_{a:bc}[Y] &= -\frac{1}{n-1} R_{bca}{}^d (\delta Y)_{d}
		+ \frac{4\Lambda}{(n-1)^2 (n-2)} g_{a[b} (\delta Y)_{c]} , \\
	T^7_{a:bc}[Y] &= \nabla_a R_{bc}{}^{de} Y_{de} , \\
	T^8_{a:bc}[Y] &= \Lambda \CKY_{a:bc}[Y] .
\end{align}

The following is a basis of the possible third total order covariant
differential operators of type $\yd{1,1}\,Y \to \yd{1,1,1}\,\hat{T}$:
\begin{align} \label{eq:hatT-def}
	\hat{T}^1_{abc}[Y]
	&= \square (\d Y)_{abc} , \\
	\hat{T}^2_{abc}[Y]
	&= R_{ab}{}^{de} \CKY_{c:de}[Y]
	 + R_{bc}{}^{de} \CKY_{a:de}[Y]
	 + R_{ca}{}^{de} \CKY_{b:de}[Y] , \\
	\hat{T}^3_{abc}[Y] 
	&= R_{ab}{}^{de} (\d Y)_{c de}
	 + R_{bc}{}^{de} (\d Y)_{a de}
	 + R_{ca}{}^{de} (\d Y)_{b de} , \\
	\hat{T}^4_{abc}[Y] 
	&= \Lambda (\d Y)_{abc} .
\end{align}

\section{Generalized normal-hyperbolicity} \label{app:hyp-conds}

In this appendix, we find the necessary and sufficient conditions for
which the generic $P$ and $Q$ operators from
identity~\eqref{eq:gen-ccyk} are generalized normally hyperbolic. Our
strategy is to first pick a differential order, say $k$, and to
parametrize the most general covariant ansatz for the principal symbols
of the potential adjugate operators $P'$ and $Q'$ at that order. Then
one can check whether the adjugate identity~\eqref{eq:adjugate} could be
satisfied at that order.

Before proceeding, in addition to the second order $P^i$, $\hat{P}^i$
and $Q^i$ operators introduced in Section~\ref{sec:ccyk} and explicitly
defined in Appendix~\ref{app:tensors}, we also need to define a fourth
order operator
\begin{dmath}
	P^9_{a\colon bc}[C]
	= \nabla^{(e}\nabla^{f)} (\nabla_{(a}\nabla_{b)} C_{e:cf} - \nabla_{(a}\nabla_{c)} C_{e:bf})
	\\
	+ \frac{1}{n-1} \square \nabla^{(d}\nabla^{e)} (g_{ab} C_{d:ec} - g_{ac} C_{d:eb}) .
\end{dmath}

Since $P$ and $Q$ are of second order, they act (by pre-composition and
up to lower-order terms) as a linear map between the spaces of principal
symbols of order $k$ and $k+2$. As can be seen from the following tensor
product decomposition table (cf.~Remark~\ref{rmk:littlewood} and the
explanations of the tables in Section~\ref{sec:ccyk})
\begin{center}
\begin{tabular}{@{}LL|CCCCCC@{}}
	\toprule
	&& \square & \nabla^2 & \nabla^4 & \nabla^6 & \nabla^8 & \cdots \\
	&& \mathbb{R} & \yd{2} & \yd{4} & \yd{6} & \yd{8} & \cdots
		\\[0.5ex] \hline &&&&& \\[-2ex]
	Y_{ab} & \yd{1,1} &
		\yd{1,1}\, Q^1 & \yd{1,1}\, Q^2 &   &   &   & \cdots
	\\
	C_{a:bc} & \yd{2,1} &
		\yd{2,1}\, P^1 & \yd{2,1}\, P^{2,3} + \yd{1,1,1}\, \hat{P}^5 &
		\yd{2,1}\, P^9 &   &   & \cdots
	\\
	\Xi_{abc} & \yd{1,1,1} &
		\yd{1,1,1}\, \hat{P}^1 & \yd{1,1,1}\, \hat{P}^2 + \yd{2,1}\, P^7 &
		  &   &   & \cdots
	\\ \bottomrule
\end{tabular} \:,
\end{center}
there is an order $k$ ($k=2$ for $Q$, and $k=4$ for $P$) after which
there are essentially no new principal symbols, meaning that for $k'\ge
k$ all operators with independent principal symbols of order $k'+2$ can
be obtained by acting with $\square$ on the operators of order $k'$.
Therefore, starting at order $k$, the pre-composition action of say $P$
on the space of potential adjugate operators $P'$ can be represented by
a square matrix. If this matrix is invertible, then the adjugate
identity~\eqref{eq:adjugate} can be satisfied, which shows generalized
normal hyperbolicity of $P$. If this matrix is singular, then it has a
right null-vector, which parametrizes an operator $P'$ such that
$P'\circ P = 0 + \text{l.o.t}$. But then, by
Lemma~\ref{lem:not-normhyp}, $P$ cannot be generalized
normally hyperbolic. The same argument works for $Q$.

In the next two Lemmas, we record the results of these calculations for
$P$ and $Q$ from~\eqref{eq:gen-ccyk}.

\begin{lemma} \label{lem:hyp-Q-ccyk}
An operator of the form
\begin{equation}
	Q_{ab}[Y] = s Q^1_{ab}[Y] + (s-t) Q^2[Y]_{ab} + \text{l.o.t} ,
\end{equation}
is generalized normally hyperbolic iff
\begin{equation}
	s \ne 0, \quad t \ne 0 ,
\end{equation}
due to the adjugate identity
\begin{equation}
	\left[\frac{1}{s} Q^1 + \left(\frac{1}{s}-\frac{1}{t}\right) Q^2\right] \circ Q
		= \square^2 + \text{l.o.t} .
\end{equation}
\end{lemma}

\begin{lemma} \label{lem:hyp-P-ccyk}
An operator of the form
\begin{equation}
	P = \begin{bmatrix}
		u P^1 + (v-u)\frac{1}{2} P^2 + u\frac{1}{2} P^3 & 0 \\
		0 & q \hat{P}^1 + q\frac{1}{2} \hat{P}^2
	\end{bmatrix}
	+ \begin{bmatrix}
		w P^3 & x P^7 \\
		y \hat{P}^5 & z\frac{1}{2} \hat{P}^2
	\end{bmatrix}
\end{equation}
is generalized normally-hyperbolic iff
\begin{equation}
\begin{gathered}
	u\ne 0, \quad
	v\ne 0, \quad
	q\ne 0, \\
	xy-wz \ne 0, \quad
	\text{and} \quad
	(n-2)(6w-v+2u)-2u\ne 0 ,
\end{gathered}
\end{equation}
due to the adjugate identity
\begin{equation}
	P' \circ P = \begin{bmatrix} \square^3 & 0 \\ 0 & \square^3 \end{bmatrix}
		+ \text{l.o.t}
\end{equation}
for
\begin{multline} \small
	P' = \begin{bmatrix}
		\square \left(u' P^1 + (v'-u')\frac{1}{2} P^2 + u'\frac{1}{2} P^3\right)
			+ (p'-\frac{3}{4}w') P^9 & 0 \\
		0 & \square \left(q' \hat{P}^1 + q'\frac{1}{2} \hat{P}^2\right)
	\end{bmatrix}
	\\
	+ \square \begin{bmatrix}
		w'\frac{1}{4} P^3 & x'\frac{1}{2} P^7 \\
		y'\frac{1}{2} \hat{P}^5 & z'\frac{1}{2} \hat{P}^2
	\end{bmatrix} ,
\end{multline}
with
\begin{gather}
	u' = \frac{1}{u} , \quad
	v' = \frac{1}{v} , \quad
	q' = \frac{1}{q} ,
	\\
	\begin{bmatrix}
		w' & x' \\
		y' & z'
	\end{bmatrix}
	= \frac{1}{xy-wz} \begin{bmatrix}
			-z & x \\
			y & -w
	\end{bmatrix}
	= \begin{bmatrix}
		w & x \\
		y & z
	\end{bmatrix}^{-1} , \\
	p'
	= \frac{(n-3)[2(u-v)^2 + 3(uv + 2uw - 4vw)] + 3u(5v+2w)}
			{2 u v [(n-2)(6w-v+2u)-2u]} .
\end{gather}
\end{lemma}

\bibliographystyle{utphys-alpha}
\bibliography{kid}

\end{document}